\documentclass{article}
\usepackage{amsmath}
\usepackage{amsthm}
\usepackage{amssymb}
\usepackage{tikz}

\newtheorem{theorem}{Theorem}
\newtheorem{proposition}[theorem]{Proposition}
\newtheorem{lemma}[theorem]{Lemma}
\newtheorem{corollary}[theorem]{Corollary}
\newtheorem{definition}[theorem]{Definition}

\newtheorem{observation}[theorem]{Observation}
\newtheorem{remark}[theorem]{Remark}
\newtheorem{example}[theorem]{Example}

\def\cB{\mathcal{B}}
\def\cC{\mathcal{C}}
\def\cD{\mathcal{D}}
\def\cG{\mathcal{G}}
\def\cH{\mathcal{H}}
\def\cK{\mathcal{K}}
\def\cL{\mathcal{L}}
\def\cQ{\mathcal{C}}
\def\cS{\mathcal{S}}
\def\cT{\mathcal{T}}
\def\cW{\mathcal{W}}
\def\cX{\mathcal{X}}
\def\KG{\mathrm{KG}}
\def\KGG{\mathcal{KG}}
\def\HC{\mathrm{HC}}
\def\HCC{\mathcal{HC}}
\def\hom{\mathsf{hom}}
\def\HOM{\mathsf{HOM}}
\def\map{\mathsf{map}}
\def\MAP{\mathsf{MAP}}
\def\HCC{\mathcal{HC}}

\def\zG{\mathbb{G}}
\def\zH{\mathbb{H}}

\def\zS{\mathbb{S}}
\def\zT{\mathbb{T}}
\def\cW{\mathbb{W}}
\def\zX{\mathbb{X}}

\let\sse=\subseteq
\def\vc#1#2{{#1}_1\zd{#1}_{#2}}
\def\zd{,\dots,}
\let\vec=\overrightarrow
\let\vf=\varphi
\let\gm=\gamma
\let\sg=\sigma
\def\All{\mathsf{All}}
\def\size{\mathsf{size}}
\def\parr{\mathsf{par}}
\def\bC{\mathbf{C}}
\def\poly{\mathrm{poly}}

\begin{document}

\title{Counting homomorphisms in plain exponential time}

\author{Amineh Dadsetan\thanks{amineh.dadsetan@gmail.com}\ \  
and Andrei A. Bulatov\thanks{abulatov@sfu.ca}\thanks{This work 
was supported by an NSERC Discovery grant}}


\maketitle

\begin{abstract}
In the counting Graph Homomorphism problem ({\sc \#GraphHom})
the question is: Given graphs $G,H$, find the number of homomorphisms
from $G$ to $H$. This problem is generally \#P-complete, moreover, 
Cygan et al.\ proved that unless the ETH is false there is no algorithm
that solves this problem in time $O(|V(H)|^{o(|V(G)|)}$. This, however,
does not rule out the possibility that faster algorithms exist for restricted
problems of this kind. Wahlstr\"om proved that {\sc \#GraphHom} can be
solved in plain exponential time, that is, in time 
$k^{|V(G)|}\poly(|V(H)|,|V(G)|)$ provided $H$ has clique 
width $k$. We generalize this result to a larger class of graphs, and 
also identify several other graph classes that admit a plain exponential 
algorithm for {\sc \#GraphHom}.
\end{abstract}

\section{Introduction}

The Exponential Time Hypothesis (ETH) 
\cite{Impagliazzo99:complexity} essentially suggests that the 
{\sc Satisfiability} problem does not admit an algorithm that is 
significantly faster than the straightforward brute force algorithm.
The ETH has been widely used to obtain (conditional) lower bounds on 
the complexity of various problems, see 
\cite{Lokshtanov11:lower} for a fairly recent survey. It however
does not forbid nontrivial algorithms for many other hard problems.

One of such problems is the {\sc Graph Homomorphism} problem 
({\sc GraphHOM} for short). A homomorphism from a graph $G$
to a graph $H$ is a mapping $\vf\colon V(G)\to V(H)$ such that 
for any edge $ab\in E(G)$ the pair $\vf(a)\vf(b)$ is an edge of $H$.
{\sc GraphHOM} asks, given graphs $G$ and $H$, whether or not 
there exists a homomorphism from $G$ to $H$ 
\cite{Hell04:homomorphism}. In the counting version 
of this problem, denoted {\sc \#GraphHOM}, the goal is to 
find the number 
of homomorphisms from $G$ to $H$. These two problems can be
solved just by checking all possible mappings from a given graph $G$
to a given graph $H$, which takes time $O^*(|V(H)|^{|V(G)|})$, 
where $O^*$ denotes asymptotics up to a polynomial factor. 
Assuming the ETH Cygan et al.\ \cite{Cygan16:tight} proved that 
the general {\sc GraphHom} and therefore {\sc \#GraphHom} 
cannot be solved in time $|V(H)|^{o(|V(G)|)}$. Related hardness 
results have also been obtained in \cite{Traxler08:time,Chen06:strong}.

In spite of this result, there are several ways to restrict 
{\sc GraphHom} which
sometimes result in a problem admitting a faster algorithm. For 
graph classes $\cG,\cH$, {\sc GraphHom}$(\cG,\cH)$ denotes
the problem {\sc GraphHom} in which the input graphs $G,H$ 
belong to $\cG,\cH$, respectively. {\sc \#GraphHOM} can be restricted
in the same way.  Both problems have received much attention 
in their own rights and as a special case of a more general
Constraint Satisfaction Problem, and much is known about their 
computational complexity. In particular, it is known when 
{\sc GraphHom}$(-,\cH)$ \cite{Hell90:h-coloring} and 
{\sc \#GraphHom}$(-,\cH)$ \cite{Dyer00:complexity} are solvable 
in polynomial time and when they are NP- or \#P-complete. 
Symbol $-$ here means that an input graph is not restricted. Similarly,
it is known when {\sc GraphHom}$(\cG,-)$ \cite{Grohe07:other} and 
{\sc \#GraphHom}$(\cG,-)$ are solvable in polynomial time.

Here we are interested in such restrictions that give rise to 
problems solvable still in exponential time but much faster than 
brute force. Specifically, {\sc GraphHom}$(\cG,\cH)$ or 
{\sc \#GraphHom}$(\cG,\cH)$ is said to be solvable in 
\emph{plain exponential} time if there is a solution algorithm 
running in time $O^*(c^{|V(G)|+|V(H)|})$, where $c$ is a 
constant. If the problem {\sc \#GraphHom}$(-,\cH)$ is
solvable in plain exponential time, the class $\cH$ is said to 
be a \emph{plain exponential class}.

The most well known plain exponential class of graphs is $\cK$, 
the class of all cliques. Note that {\sc \#GraphHom}$(-,\cK)$ 
is equivalent to the {\sc \#Graph Colouring} problem, in which the 
problem is, given a graph $G$ and a number $k$, to find the 
number of $k$-colourings of $G$. A fairly straightforward dynamic 
programming algorithm solves this problem in time 
$O^*(3^{|V(G)|})$; we outline this algorithm in 
Example~\ref{exa:plain-expo-cw}. A more sophisticated algorithm
\cite{Koivisto06:algorithm} solves this problem in time 
$O^*(2^{|V(G)|})$. If $\cH$ is a class of graphs of tree width $k$ 
then {\sc \#GraphHom}$(-,\cH)$ is solvable in time 
$O^*((k+3)^{|V(G)|})$, see, \cite{Fomin07:exact}. For the class 
$\cD_c$ of graphs of degree at most $c$ the problems 
{\sc \#GraphHom}$(\cD_c)$ and {\sc \#GraphHom}$(\cD_c)$ 
can be solved in time $O^*(c^{|V(G)|})$ by a minor modification of 
the brute force enumeration algorithm. Finally, Wahlstr\"om 
\cite{Wahlstrom09:new}
obtained probably the most general result so far on plain exponential 
graph classes, proving that if $\cH$ only contains graphs of clique 
width $k$ then {\sc \#GraphHom}$(-,\cH)$ can be solved in time
$O^*((2k+1)^{|V(G)|+|V(H)|})$. 
The algorithm from \cite{Wahlstrom09:new}
is also dynamic programming and uses the representation of
(labeled) graphs of bounded clique width through a sequence of 
operations such as disjoint union, connecting vertices with
certain labels, and relabelling vertices (\emph{$k$-expressions}). 

In this paper we further expand the class of graphs for which
plain exponential counting algorithms are possible by adding 
one more operation to the construction of graphs of bounded
clique width. The new class of graphs includes families of 
graphs of unbounded clique width, for  instance, hypercubes,
and therefore is strictly larger than the class of graphs of 
bounded clique width. By means of this new set of operations 
one can define a new graph `width' measure that we call
\emph{extended clique width}. Graphs of extended clique 
width at most $k$ can also be represented by 
\emph{extended $k$-expressions}. Let $\cX_k$ denote the
class of graph of extended clique width at most $k$. 

We then show that given an arbitrary graph $G$, a graph 
$H$ of extended clique width $k$, and an extended 
$k$-expression $\Phi$ representing $H$, the number $\hom(G,H)$
of homomorphisms from $G$ to $H$ can be found in time
$O^*((2k)^{|G|})$. Similar to \cite{Wahlstrom09:new}, the 
algorithm is dynamic programming and iteratively computes 
numbers $\hom(G',H')$, where $G'$ is an induced subgraph
of $G$ and $H'$ is a graph represented by a subexpression of 
$\Phi$. Clearly, as one cannot assume that an extended 
$k$-expression representing $H$ is known in advance, 
this algorithm alone does not guarantee that $\cX_k$ is
plain exponential. However, we also show that given a graph
$H$ of extended width at most $k$, an extended $k$-expression
representing $H$ can be found in time $O^*((2k)^{|H|})$. 
Combined with the previous result we thus obtain the following

\begin{theorem}\label{the:main-intro}
For any fixed $k$ the class $\cX_k$ is plain exponential.
\end{theorem}

Apart from graphs of bounded extended clique width we identify
two less general plain exponential classes of graphs. The first one 
consists of subdivisions of 
cliques: Let $\cS$ be a class of graphs, then $\cK(\cS)$ denotes
the class of graphs $H$ obtained as follows. Take $H'\in\cS$, a 
clique on vertices $\{\vc vn\}$, and for any edge $v_iv_j$ of
the clique, $i\ne j$, replace this edge with a copy of $H'$, that is,
connect $v_i,v_j$ to all vertices of $H'$ and include all the edges 
of $H'$.

\begin{theorem}\label{the:subdivision}
For any plain exponential class $\cS$ of graphs, the class $\cK(\cS)$
is also plain exponential.
\end{theorem}

The second class consists of well studied Kneser graphs: $\KGG_k$ 
is the class of graphs, whose vertices are the $k$-element subsets 
of a certain set, and two vertices are connected if and only if the 
corresponding subsets are disjoint. 

\begin{theorem}\label{the:Kneser}
For every $k$ the class $\KGG_k$ is plain exponential.
\end{theorem}

\section{(Extended) Clique width}\label{sec:cliquewidth}

\subsection{Homomorphisms, plain exponential time}

As always we denote the vertex set of a graph $G$ by $V(G)$, 
and its edge set by $E(G)$. A \emph{homomorphism} of a graph 
$G$ to a graph $H$ is a mapping $\vf\colon V(G)\to V(H)$ 
such that $\vf(u)\vf(v)\in E(H)$ for any $uv\in E(G)$.
The \emph{Counting Graph Homomorphism} problem 
{\sc \#GraphHom} is defined as follows: given graphs $G,H$, 
find the number of homomorphisms from $G$ to $H$. Its decision 
version --- does there exist a homomorphism from $G$ to $H$ --- is 
denoted by {\sc GraphHom}. Graph homomorphisms and the
related combinatorial problems have been extensively studied 
\cite{Hell04:homomorphism}.
If $H$ is allowed only from a class $\cH$ of graphs, the resulting
counting and decision problems are denoted 
{\sc \#GraphHom$(-,\cH)$} 
and {\sc GraphHom$(-,\cH)$}, respectively. 

We will be concerned with the complexity and the best running time 
of algorithms for {\sc \#GraphHom$(-,\cH)$}. In particular, we say 
that a class $\cH$ of graphs is \emph{plain exponential} if there is 
an algorithm that solves the problem {\sc \#GraphHom$(-,\cH)$} in 
\emph{plain exponential time}: there exists a constant $c$ such 
that on input $G,H$, $H\in\cH$, the algorithm runs in time 
$O^*(c^{|V(G)|+|V(H)|})$, where $O^*$ means asymptotics
up to a factor polynomial in $|V(G)|,|V(H)|$. Note that we will 
always assume that $G$ and $H$ are connected, since otherwise the 
existence or the number of homomorphisms from $G$ to $H$ can 
be deduced from those of their connected components.

\begin{example}\label{exa:plain-expo-h-col}\rm
({\sc $H$-Colouring}.)
If $\cH$ consists of just one graph, $H$, the problems \linebreak
{\sc \#GraphHom$(-,\cH)$}, {\sc GraphHom$(-,\cH)$} are known 
as {\sc \#$H$-Colouring} and {\sc $H$-Colouring},
respectively. The {\sc \#$H$-Colouring} problem is solvable in 
polynomial time if $H$ is a complete graph with all loops present, 
or is a complete bipartite graph \cite{Dyer00:complexity}.
The {\sc $H$-Colouring} problem is solvable in polynomial time if $H$ 
contains a loop or is bipartite \cite{Hell90:h-coloring}. Otherwise these problems 
are \#P- and NP-complete, respectively. Since the brute force 
algorithm for this problems runs in $O(|V(H)|^{|V(G)|})$ time, 
{\sc \#$H$-Colouring} and {\sc $H$-Colouring} are always solvable 
in plain exponential time. Also, by inspecting the solution algorithms
from \cite{Dyer00:complexity,Hell90:h-coloring} these results can be slightly generalized: 
{\sc \#GraphHom$(-,\cH)$} is solvable in polynomial time whenever 
every graph from $\cH$ is a complete 
graph with all loops, or a complete bipartite graph. Similarly 
{\sc GraphHom$(-,\cH)$} is polynomial time solvable if every graph 
from $H$ contains a loop or is bipartite.
\end{example}

\begin{example}\label{exa:plain-expo-bound-degree}\rm
(Graphs of bounded degree.)
As is mentioned in the introduction, if the degrees of graphs from 
$\cH$ are bounded by a number $c$, the (improved) brute force 
algorithm solves {\sc \#GraphHom$(-,\cH)$},
{\sc GraphHom$(-,\cH)$}. Let $G,H$ be input graphs, 
$H\in\cH$. We assume 
$G$ is connected; otherwise the procedure below has to be 
performed for each connected component, and the results 
multiplied. Order the vertices $v_1,\dots,v_n$ of $G$ in such a 
way that each vertex except for the first one is adjacent to one of the 
preceding vertices. Then the brute force algorithm is organized as 
follows: Assign images to $v_1,\dots,v_n$ in turn. There are $|H|$ 
possibilities to map $v_1$, but then if $v_i$ is adjacent to $v_j$, 
$j<i$, the image of $v_j$ is fixed, and therefore there are at most 
$c$ possibilities for the image of $v_i$. Thus, the algorithm
runs in $O^*(c^n)$. This approach also allows $H$ to have 
bounded number of vertices of high degree.
\end{example}

\begin{example}\label{exa:plain-expo-cw}\rm
(Graphs of bounded clique width.)
Let $\cQ_k$ denote the class of all graphs of clique width at most 
$k$ (to be defined in Section~\ref{sec:clique-width}). Then 
{\sc \#GraphHom$(-,\cQ_k)$}, {\sc GraphHom$(-,\cQ_k)$} can be 
solved in time $O^*((2k+1)^{|V(G)|+|V(H)|})$, implying that 
$\cQ_k$ is plain exponential \cite{Wahlstrom09:new}. 

Here we briefly describe the simple algorithm solving {\sc 
\#GraphHom$(-,\cK)$}, where $\cK$ is the class of cliques. 
Given a graph $G$ and a number $s$ (or, equivalently, the clique 
$K_s$) 
the solution algorithm maintains an array $N(S,\ell)$ for $S\sse V$ and 
$\ell\le s$, which contains the number of homomorphisms from
the subgraph of $G$ induced by $S$ to an $\ell$-element clique.
To compute each $N(S,\ell)$ we go over all subsets $S'\sse S$,
consider the vertices from $S'$ to be mapped to the $\ell$-th
vertex of the $\ell$-clique. Then there are $N(S-S',\ell-1)$ ways
to map the remaining vertices, and $N(S,\ell)$ is the sum of all
numbers like this. It is not hard to see that the running time of
this algorithm is $O^*(3^{|V(G)|}$. It can be improved to run 
in time $O^*(2^{|G|})$ \cite{Koivisto06:algorithm}, and some 
further improvements are possible in certain cases 
\cite{Fomin07:improved}.
\end{example}

We will often deal with vertex labeled graphs. It will be convenient to 
represent labels on vertices of a graph $G$ as a \emph{label function} 
$\pi\colon V(G)\to[k]$ ($[k]=\{1\zd k\}$), in which case we say that 
$G$ is \emph{$k$-labeled}. Graph $G=(V,E)$ equipped with a label 
function $\pi$ will be denoted by $\zG=(V,E,\pi)$. The $k$-labeled 
graph $\zG$ is then called a \emph{$k$-labelling} of $G$. Let 
$\zG_1=(V_1,E_1,\pi_1)$ and $\zG_2=(V_2,E_2,\pi_2)$  are 
$k$-labeled graph. A mapping $\vf\colon V_1\to V_2$ is a 
\emph{homomorphism} of $k$-labeled graph $\zG_1$ to 
$k$-labeled graph 
$\zG_2$ if it is a homomorphism of graph $G_1=(V_1,E_1)$ to
$G_2=(V_2,E_2)$ respecting the labelling, that is, 
$\pi_2(\vf(v))=\pi_1(v)$ for every $v\in V_1$. 

The following 
notation will also be useful. Let again $\zG_1,\zG_2$ be
$k$-labeled graphs, such that $V_1,V_2$ are disjoint. Then 
$\zG_1\bigoplus\zG_2=(V_1\uplus V_2,E_1\uplus E_2,
\pi_1\uplus\pi_2)$, where 
\[
\pi_1\uplus\pi_2(v) = 
\begin{cases}
\pi_1(v), &\textit{if $v \in V_1$},\\
\pi_2(v), &\textit{if $v \in V_2$.}
\end{cases}
\]

Finally, the subgraph of a graph $G=(V,E)$ induced by a set $S\sse V$ 
is denoted by $G[S]$. For a $k$-labeled graph $\zG=(V,E,\pi)$, by
$\zG[S]$ we denote the $k$-labeled subgraph induced by $S\sse V$.
Note that the labelling function of $\zG[S]$ is $\pi\vert_S$, i.e., 
the restriction of $\pi$ on the set~$S$.

\subsection{Clique width and $k$-expressions}\label{sec:clique-width}

The simplest way to introduce clique width of a graph is through 
$k$-expressions.

\begin{definition}
The following operators are defined on $k$-labeled graphs.
\begin{itemize}
\item 
$\cdot_i$: Construct a graph with one vertex, which is labeled $i\in[k]$. 
\item 
$\rho_{i \rightarrow j}(\zG)$: Relabel all vertices with label $i\in[k]$ 
of a $k$-labeled graph $\zG$ to label $j\in[k]$. 
\item 
$\eta_{ij}(\zG)$, for $i \neq j$: Add edges from every vertex labeled 
$i$ to every vertex labeled $j$ in $\zG$, i.e. add edges $uv$ for any 
vertices $u, v$ where $u$ has label $i$ and $v$ has label $j$. 
\item 
$\zG_1 \bigoplus \zG_2$: The disjoint union of $k$-labeled graphs 
$\zG_1$ and $\zG_2$. 
\end{itemize}
A \emph{$k$-expression} is any (properly formed) formula using the 
above operators. 

Every $k$-expression represents a $k$-labeled graph. We say that a
graph $G=(V,E)$ is represented by $k$-expression $\Phi$, if 
there exists a $k$-labelling $\pi$ of the vertices of $G$ such that 
$\Phi$ represents $\zG=(V,E,\pi)$.
A graph 
has \emph{clique width} $k$ if it is represented by a $k$-expression.
The class of all graphs of clique width $k$ is denoted by $\cQ_k$.
\end{definition}

Wahlstr{\"{o}}m in \cite{Wahlstrom09:new} used $k$-expressions 
of graphs to show that $\cQ_k$ is plain exponential. However, 
$k$-expressions suitable for his plain exponential algorithm must
satisfy an extra condition. Let $\Phi$ be a $k$-expression representing 
a $k$-labeled graph $\zG$. Note that any subformula of $\Phi$ 
represents a subgraph of $\zG$.
We say that $k$-expression $\Phi$ is \emph{safe} if for every its
subexpression $\Phi_1 \bigoplus \Phi_2$ such that $\Phi_1,\Phi_2$ 
represent graphs $\zG_1,\zG_2$, respectively, the graph 
$\zG_1$ equals $\zG[V(\zG_i)]$ for $i=1,2$. In other words all 
edges of $\zG$ between vertices of 
$\zG_i$, $i=1,2$, are already edges of $\zG_i$. 

\begin{lemma}[\cite{Wahlstrom09:new}]%
\label{lem:safe-expression}
(1) Every graph of clique width $k$ can be represented by a 
safe $k$-expression.\\[2mm]
(2) A safe $k$-expression for a graph of clique width $k$
can be found in plain exponential time.
\end{lemma}

\subsection{Extended $k$-expressions}

In this section we introduce a more general version of $k$-expressions,
and accordingly a more general version of clique width.

Fix a natural $k$. By $\vec{n}$ we denote a vector 
$(\vc nk) \in ([k] \cup \{0\})^k$. For such a vector 
$\vec{n}$, let
\[
\mathcal{L}(\vec{n})=\{(i_1, j_1, i_2, j_2)\mid  
i_1, i_2 \in [k], 1 \leq j_1 \leq n_{i_1}, 1 \leq j_2 \leq n_{i_2}\}.
\]

New $k$-expressions require two more operators on $k$-labeled 
graphs. The first one does not have analogues in $k$-expressions.
\begin{definition}\label{def:operator-B}
Let  $\vec{n} =(\vc nk) \in ([k] \cup \{0\})^k$, 
$\sigma: [k] \to [k]$, and $\cS \subseteq \cL(\vec{n})$. Also, let 
$\cS$ be a symmetric set, that is, if $(i_1, j_1, i_2, j_2) \in \cS$ 
then $(i_2, j_2, i_1, j_1) \in \cS$. Operator 
$\beta_{\vec{n}, \sigma, \cS}$ transforms 
$k$-labeled graph $\zG_1=(V_1, E_1, \pi_1)$ to a $k$-labeled graph 
$\zG_2=(V_2, E_2,\pi_2)$ as follows:
\begin{itemize}
\item 
$V_2 = \bigcup_{i = 1}^{k} C_i$, where 
$C_i = \{a_{j}| j \in \{0, ..., n_i\},\text{ $a\in V_1$ 
and } \pi_1(a) = i\}$. The vertices of $a_0$, $a\in V_1$, are called 
original vertices of $\zG_2=\beta_{\vec{n}, \sigma, \cS}(\zG_1)$ 
and are identified with their corresponding vertices from $V_1$;
\item 
$(a_{j}, b_{j'})\in E_2$ if and only if
$(a, b)\in E_1$, and $(\pi_1(a), j,\pi_1(b) ,j') \in\cS$ or  
$j = j' = 0$;
\item 
$\pi_2(a_{j}) = \begin{cases}
\pi_1(a), & \textit{if $j = 0$,}\\
\sigma(\pi_1(a)), &\text{otherwise.}
\end{cases}$
\end{itemize}
We also refer to this operator as the \emph{beta operator}. 
\end{definition}

The second operator combines disjoint union with a sequence of 
adding edges operators. 

\begin{definition}\label{def:eta-operator}
Let $\cT \subseteq [k] \times [k]$. Operator $\eta_{\cT}$ takes two 
$k$-labeled graphs as input and produces a $k$-labeled graph 
as output. For $k$-labeled graphs $\zG_1=(V_1,E_1,\pi_1)$, and 
$\zG_2=(V_2,E_2,\pi_2)$, $V_1,V_2$ disjoint the $k$-labeled graph  
$\eta_{\cT}(\zG_1,\zG_2) = (V, E, \pi)$, is defined as follows:
\begin{align*}
V &= V_1 \cup V_2\\
E &=E_1\cup E_2\cup\{(a, b)\mid a\in V_1, b\in V_2, 
\pi_1(a)=i, \pi_2(b)=j,\ (i,j)\in\cT \}\\
\pi &= 
\begin{cases}
\pi_1(a),&   \text{if $a \in V_1$,}\\
\pi_2(a), &   \text{if $a \in V_2$.}
\end{cases}
\end{align*}
We also refer to this operator as the \emph{connect operator}. 
\end{definition}

An \emph{extended $k$-expression} is a (properly formed) 
expression that involves operators $\cdot_i$ ($i\in[k]$), 
$\rho_{i\to j}$ ($i,j\in[k]$), $\beta_{\vec{n}, \sigma, \cS}$, and 
$\eta_{\cT}$, where 
$\vec n,\sigma,\cS,\cT$ are as in Definitions~\ref{def:operator-B},
\ref{def:eta-operator}. Similar to $k$-expressions, extended
$k$-expressions represent $k$ labeled graphs, as well as usual 
graphs. Next we explore what kind of graphs and $k$-labeled 
graphs can be represented by extended $k$-expressions.

Note that if $\zG_1$ and $\zG_2$ are two isomorphic 
$k$-labeled graphs, and $\zG_1$ has an extended $k$-expression 
$\Phi$, then $\Phi$ is an extended $k$-expression for $\zG_2$ 
as well.

As is easily seen, the connect operator can be expressed through 
disjoint union and adding edges. However, we will need properties
similar to the safety of $k$-expressions. Unfortunately, the 
beta operator does not allow an equally clean and easy definition
of safety, as in the case or $k$ expressions, and we use the 
connect operator instead.

Let $\zG = \eta_{\zT}(\zG_1, \zG_2)$. It is straightforward from the 
definition that  
$\zG[V(\zG_1)]$ is equal to $\zG_1$ and $\zG[V(\zG_2)]$ is equal to 
$\zG_2$, that is, $\eta_\cT$ does not add edges inside $\zG_1,\zG_2$. 
Similarly, if $\zG= \beta_{\vec{n}, \zS, \sigma}(\zG_1)$, 
then again $\zG[V(\zG_1)]$ is equal to $\zG_1$. Similar to 
$k$-expressions we say that an extended $k$-expression $\Phi$ is
\emph{safe} if for every its subexpressions $\eta_\cT(\Phi_1,\Phi_2)$ 
and $\beta_{\vec n,\sigma,\cS}(\zG_1)$ such that $\Phi_1,\Phi_2$ 
represent graphs $\zG_1,\zG_2$, 
respectively, it holds $\zG_1=\zG[V(\zG_i)]$ for $i=1,2$. 
The following property is straightforward.

\begin{lemma}\label{lem:extended-safe}
Any extended $k$ expression is safe.
\end{lemma}

For an extended $k$-expression $\Phi$, 
$\size(\Phi)$ denotes the total number of operands, connect operators, 
beta operators, and the maximal subsequences of relabelling 
operators of $\Phi$. 

A graph $G=(V,E)$ is said to have \emph{extended clique width $k$} 
if there is a $k$-labelling $\pi$ of $G$ and an extended 
$k$-expression $\Phi$ that represents $\zG=(V,E,\pi)$. If such a
$\pi $ exists we also say that $\Phi$ represents $G$. The class of
all graphs of extended clique width is denoted by $\cX_k$.

We complete this section showing that $\cQ_k$ is a subset
of $\cX_k$. 

\begin{proposition}\label{equivalent}
Any graph $G$ that can be represented by a $k$-expression, can also 
be represented by an extended $k$-expression.
\end{proposition}

\begin{proof}
We start with a piece of terminology. For a sequence $\Psi$ of 
operators of the form $\rho_{i\to j}$ and $\eta_{ij}$, consider the
$k$-labeled graph $\zG=\Psi(\zG_1 \bigoplus \zG_2)$ for some 
$\zG_1,\zG_2$. Expression $\Psi$ is said to \emph{connect} a 
vertex $a$ to a vertex $b$, if $ab$ is an edge of $\zG$, but not 
of $\zG_1\bigoplus\zG_2$, that is, if there is an operator 
$\eta_{ij}$ in $\Psi$ that connects $a$ to $b$. Expression 
$\Psi$ is said to relabel a vertex $a$ with label $i$ to label $j$, 
if $a$ has label $i$ in $\zG_1\bigoplus\zG_2$ and label $j$ in $\zG$. 
Also, $\hat{\Psi}$ denotes the sequence of operators that is 
obtained from $\Psi$ by removing all the $\eta_{ij}$ operators..

Let $\Phi$ be a $k$-expression representing graph $G$. By 
Lemma~\ref{lem:safe-expression} $\Phi$ can be assumed safe. 
We proceed by induction on the structure of $\Phi$. If $G$ is a 
graph with one vertex, there is nothing to prove.
If $G$ has more than one vertex, we can write $\Phi$ as  
$\Psi(\zG_1 \bigoplus \zG_2)$ where,  $\zG_1$ and $\zG_2$ 
are represented by some subexpressions of $\Phi$ and $\Psi$ 
is a sequence of operators of the form $\rho_{i\to j}$ and $\eta_{ij}$. 
Let $\zG_1=(V_1, E_1, \pi_1)$ and $\zG_2=(V_2, E_2, \pi_2)$. 

As is easily seen, for any $a \in V_1$ and $b \in V_2$, if 
$ab\in E(G)$, then there is operator $\eta_{st}$ in $\Psi$ that 
connects $a$ to $b$. Conversely, if some operator $\eta_{st}$ 
in $\Psi$ connects vertex $a\in V_1$  
to a vertex $b\in V_2$ where, $\pi_1(a) = i$ and $\pi_2(b) = j$, then, 
$\Psi$ connects every vertex $x \in V_1$ to every vertex 
$y \in V_2$, with $\pi_1(x) = i$ and $\pi_2(y) = j$, because 
the vertices from $\zG_1\bigoplus\zG_2$ with the same label, 
remain with the same label, after applying any operator. 
So, there is a set $\cT$ of pairs $(i, j)$ 
such that $\Psi$ connects  every vertex $x \in V_1$ to 
every vertex $y \in V_2$ with $\pi_1(x) = i$ and $\pi_2(y) = j$.

Therefore, the set of edges between $\zG_1$ and $\zG_2$ which are 
added by $\Psi$ is the same as those which are added by 
$\eta_{\cT}$. Also, $\zG[V_1] = \zG_1$ and $\zG[V_2] = \zG_2$, as $\Phi$
is safe. Thus, $\eta_{\cT}(\zG_1, \zG_2)$ 
and $\Psi(\zG_1 \bigoplus \zG_2)$ have the same set of 
vertices and same set of edges. Moreover, $\Psi$ relabels any 
vertices  same way as $\hat{\Psi}$ does. Therefore 
$\Psi(\zG_1 \bigoplus \zG_2)$ and 
$\hat{\Psi}\eta_{\cT}(\zG_1, \zG_2)$ represent the same 
$k$-labeled graph. By the induction hypothesis the result follows.
\end{proof}

\begin{corollary}\label{cor:clique-extended-clique}
Every graph that has clique width $k$ also has extended clique 
width $k$.
\end{corollary}

Next we show that not all graphs of extended clique width $k$ 
also have clique width $k$. More precisely, we present a class of 
graph of extended clique width 2 that does not have bounded 
clique width.

\subsection{Graph class of bounded extended but not regular clique width}

An \emph{$n$-dimensional hypercube}, denoted $\HC_n$ is a graph 
whose vertices are $n$-bit binary vectors, and two vertices are adjacent 
if the Hamming distance between them is exactly~1. Let 
$\HCC=\{\HC_n\mid n\in\mathbb N\}$. We first show that each 
hypercube is represented by an extended 2-expression, and therefore 
has extended width 2. 

Extended 2-expressions for hypercubes are constructed by induction 
on the dimensionality of the hypercube. The base cases of induction are $\HC_0$ and 
$\HC_1$. An extended 2-expression for $\HC_0$ is $\cdot_1$, and 
an extended 2-expression for $\HC_1$ is 
$\eta_{\{(1, 2)\}}(\cdot_1, \cdot_2)$.

Suppose that for $m<n$ the graph $\HC_m$ has an extended 
2-expression. Let $\Phi$ be an extended 2-expression for $\HC_n$. 
Let $\vec{n} = (1, 1)$, let $\sigma:[2] \to 2$ be 
the function given by $\sigma(1) \mapsto 2, \sigma(2) \mapsto 1$, 
and let $\cS$ be
$
\{(1, 1, 1, 1), (2, 1, 2, 1),(1, 1, 2, 1)$, $(2, 1, 1, 1),
(1, 0, 2, 1), (2, 1, 1, 0), ( 2, 0, 1, 1), (1, 1, 2, 0)
\}.
$
Then it is not hard to see that $\beta_{\vec{n}, \sigma, \cS}\Phi$ 
is an extended 2-expression for $\HC_{n + 1}$.

\begin{lemma}\label{lem:hcubes}
$\HCC$ does not have bounded clique width.
\end{lemma}
 
\begin{proof}
Let $k$ be a constant, let $n$ be a sufficiently large number, 
let $V$ be a set of cardinality $n$, and, for the sake of contradiction, 
let  $\Phi_0$ be a $k$-expression for $\HC_n$. We 
define a finite sequence of $k$-expressions $\{\Phi_i\}$  
which starts with $k$-expression $\Phi_0$.

For $i \ge 0$, the $( i + 1)$-th element of the sequence is defined 
from the $i$-th element of the sequence as follows: if $\Phi_i$ 
represents a graph with more than 1 vertices, it has the form 
$\Psi_i(\Phi_{i,l} \bigoplus \Phi_{i,r})$, where 
$\Psi_i$ is a sequence of recolouring and renaming operators 
and $\Phi_{i,l}$ and $\Phi_{i,r}$ are two $k$-expressions. Then, 
let $\Phi_{i + 1}$ be the $k$-expression from $\{\Phi_{i,l},\Phi_{i,r}\}$
that represents the graph with the greatest number of vertices.
Thus, if $G_i$ is the graph represented by $\Phi_i$ and $G_{i+1}$
the graph represented by $\Phi_{i+1}$, then 
$|V(G_{i+1})| \geq |V(G_i)|/2$. 

Let $G_i$ be the first graph in the sequence such that 
$|V(G_i)|$ is less than $n$. As $G_i$ is the
first such graph, $|V(G_{i - 1})| \geq n$. By the observation above, 
$|V(G_i)|\ge n/2$.
By the Pigeonhole principle, for some  label in $[k]$, say $1$, 
there are more than $n/2k$ vertices of $G_i$ labeled with it. 
Since $n$ is sufficiently large, there are at least three vertices 
that are labeled $1$. Let us denote these vertices $a,b,c$. 
Since $G_i$ contains fewer than $n$ vertices and the degree of 
$a$ in $G$ equals $n$, there are at least two vertices outside 
$G_i$ adjacent to $a$. Denote them $d$ and $e$. Now, at 
some point in $\Phi$ vertices $d,e$ are connected to 
$a$. However, as $b,c$ have the same label as $a$, the 
operator $\eta_{j\ell}$ that connects $a$ to $d$ and $a$ to $e$, 
also adds edges $bd,be,cd,ce$. So, 
the subgraph of $\HC_n$ induced by $\{a, b, c, d,e\}$ contains a 
complete bipartite graph $K_{3,2}$. 

Let us view $a,b,c,d,e$ as $n$-bit vectors. Then $d$ differs from 
from each of $a,b,c$ in one bit, but the position of that bit is 
different, say, it is 1,2 and 3, respectively. This means that $a,b,c$
are equal in all the remaining positions. Since $e$ also differs 
from each of $a,b,c$ in one position, these must be the same 
positions, and $d=e$.
\end{proof}

\subsection{Finding an extended $k$-expression for a given graph}

Next we show how to find an extended $k$-expression for a given 
graph $G$ if it has one, in time $O^*(k^{|V(G)|})$. 

One of the ingredients of our algorithm is the problem of deciding
whether two $k$-labeled graphs are isomorphic. $k$-labeled 
graphs $\zG=(V_1,E_1,\pi_1),\zH=(V_2,E_2,\pi_2)$ are isomorphic 
if there exists an isomorphism $\vf$ from the graph $G=(V_1,E_1)$
to $H=(V_2,E_2)$ such that $\pi_1(a)=\pi_2(\vf(a))$ for $a\in V_1$. 
We show that 
this problem can be reduced to the regular {\sc Graph Isomorphism}
problem and use the celebrated result by Babai \cite{Babai16:graph} 
that there is an algorithm that, given graphs $G$ and $H$, decides 
whether there exists an isomorphism between $G$ and $H$ in time 
$O(2^{log(|V(G)|)^{O(1)}})$.  

\begin{proposition}\label{part-iso-time}
There is an algorithm that decides if $k$-labeled graphs $\zG,\zH$ are 
isomorphic and runs in time $O^*(2^{log((k + 2)n )^{O(1)}})$. 
\end{proposition}

Before proving Proposition \ref{part-iso-time}, we introduce a gadget 
construction, see, Fig~\ref{D}, which we borrow from 
\cite{Fomin15:lower} with minor modifications. 
Let us denote the  gadget in Fig~\ref{D} by $D$. It is proved in 
\cite{Fomin15:lower} that for each homomorphism 
$\vf: D \rightarrow D$ 
and $i \in [5]$, $\phi(z) = z$ and $\vf(z) \neq \vf(x_i)$.

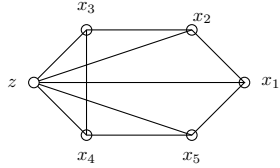
\begin{figure}[h]
\centering
\begin{tikzpicture}[scale=0.7, transform shape]
\draw (0, 1) --(1, 2) --(3, 2)-- (4, 1) --(0, 1)-- 
               (1, 0) --(3, 0) -- (4, 1);
\draw (0, 1) --(3, 2);
\draw (0, 1) --(3, 0);
\draw (1, 2) --(1, 0);
\draw [fill=white] (0, 1) circle [radius=0.1];
\node [left] at (-0.2, 1) {$z$};

\draw (1, 2) circle [radius=0.1];
\node [above] at (1, 2.2) {$x_3$};

\draw (3, 2) circle [radius=0.1];
\node [above] at (3.2, 2) {$x_2$};

\draw (3, 0) circle [radius=0.1];
\node [below] at (3, -0.2) {$x_5$};

\draw (1, 0) circle [radius=0.1];
\node [below] at (1, -0.2) {$x_4$};

\draw (4, 1) circle [radius=0.1];
\node [right] at (4.2, 1) {$x_1$};

\end{tikzpicture}
\caption{The graph $D$}
\label{D}
\end{figure}

We join $q $ more such gadgets $D$ in a row to construct a larger 
gadget $T_{q }$ that has total $q + 1$ copy of gadget $D$. 
We label the left-most vertex of the $i$-th block in the chain by 
$z_{i - 1}$ except for the first block, see, Fig~\ref{T-l}. Each   
automorphism of $T_{q }$ preserves the order on $z$'s which means 
for each $i \in [q]$ and isomorphism 
$\vf: T_{q } \rightarrow T_{q}$, $\vf(z_i) = z_i$. This property 
is also proved in \cite{Fomin15:lower}. 

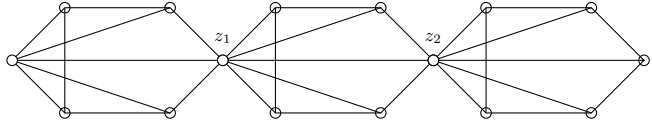
\begin{figure}[h]
\centering
\begin{tikzpicture}[scale=0.7, transform shape]
\draw (0, 1) --(1, 2) --(3, 2)-- (4, 1) --(0, 1)-- 
               (1, 0) --(3, 0) -- (4, 1);
\draw (0, 1) --(3, 2);
\draw (0, 1) --(3, 0);
\draw (1, 2) --(1, 0);
\draw [fill=white] (0, 1) circle [radius=0.1];

\draw (1, 2) circle [radius=0.1];
\draw (3, 2) circle [radius=0.1];

\draw (3, 0) circle [radius=0.1];
\draw (1, 0) circle [radius=0.1];
---------------
\draw (4, 1) --(5, 2) --(7, 2) --(8, 1) --(4, 1)-- 
               (5, 0) --(7, 0) --(8, 1);
\draw (4, 1) --(7, 2);
\draw (4, 1) --(7, 0);
\draw (5, 2) --(5, 0);
\draw (4, 1) circle [radius=0.1];

\draw (5, 2) circle [radius=0.1];
\draw (7, 2) circle [radius=0.1];
\draw (8, 1) circle [radius=0.1];
\draw (7, 0) circle [radius=0.1];
\draw (5, 0) circle [radius=0.1];

\draw [fill=white] (4, 1) circle [radius=0.1];
\node [above] at (4, 1.2) {$z_1$};

-------------------------
\draw (8, 1) --(9, 2) --(11, 2) --(12, 1) --(8, 1)-- 
               (9, 0) --(11, 0) --(12, 1);
\draw (8, 1) --(11, 2);
\draw (8, 1) --(11, 0);
\draw (9, 2) --(9, 0);
\draw (8, 1) circle [radius=0.1];

\draw (9, 2) circle [radius=0.1];
\draw (11, 2) circle [radius=0.1];
\draw (12, 1) circle [radius=0.1];
\draw (11, 0) circle [radius=0.1];
\draw (9, 0) circle [radius=0.1];

\draw [fill=white] (8, 1) circle [radius=0.1];
\node [above] at (8, 1.2) {$z_2$};
\end{tikzpicture}
\caption{Gadget $T_q$}
\label{T-l}

\end{figure}

For each $i \in [q]$, we replace  $z_i$ in $T_q$ with a copy of 
$K_{n + 3}$, a clique on $n + 3$ vertices, and connect every 
vertex of $K_{n + 3}$ to all the neighbours of $z_i$, its special vertex, 
in the next subsequent block that is $(i + 1)$th block. Note that 
only one special   vertex of the $(i + 1)$th copy of $K_{n + 3}$, 
that is called $z_i$,  is connected to some vertices of previous block. 
Denote the new graph by $T_{q, n }$. see Fig~\ref{injectclique} 

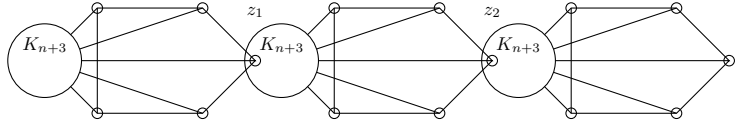
\begin{figure}[h]
\centering
\begin{tikzpicture}[scale=0.7, transform shape]
\draw (9, 1) --(10, 2) --(12, 2) --(13, 1) --(9, 1)-- 
               (10, 0) --(12, 0) --(13, 1);
\draw (9, 1) --(12, 2);
\draw (9, 1) --(12, 0);
\draw (10, 2) --(10, 0);
\draw (9, 1) circle [radius=0.1];

\draw (10, 2) circle [radius=0.1];
\draw (12, 2) circle [radius=0.1];
\draw (13, 1) circle [radius=0.1];
\draw (12, 0) circle [radius=0.1];
\draw (10, 0) circle [radius=0.1];
---------------
\draw [fill=white] (9, 1) circle [radius=0.7];
\node [above] at (9, 1) {$K_{n + 3}$};
\node [above] at (8.5, 1.7) {$z_2$};
----------------------------
\draw (4.5, 1) --(5.5, 2) --(7.5, 2) --(8.5, 1) --(4.5, 1)-- 
               (5.5, 0) --(7.5, 0) --(8.5, 1);
\draw (4.5, 1) --(7.5, 2);
\draw (4.5, 1) --(7.5, 0);
\draw (5.5, 2) --(5.5, 0);
\draw (4.5, 1) circle [radius=0.1];

\draw (5.5, 2) circle [radius=0.1];
\draw (7.5, 2) circle [radius=0.1];
\draw (8.5, 1) circle [radius=0.1];
\draw (7.5, 0) circle [radius=0.1];
\draw (5.5, 0) circle [radius=0.1];
---------------
\draw [fill=white] (4.5, 1) circle [radius=0.7];
\node [above] at (4.5, 1) {$K_{n + 3}$};
\node [above] at (4, 1.7) {$z_1$};

------------------------
\draw (0, 1) --(1, 2) --(3, 2)-- (4, 1) --(0, 1)-- 
               (1, 0) --(3, 0) -- (4, 1);
\draw (0, 1) --(3, 2);
\draw (0, 1) --(3, 0);
\draw (1, 2) --(1, 0);
\draw [fill=white] (0, 1) circle [radius=0.7];
\node [above] at (0, 1) {$K_{n + 3}$};

\draw (1, 2) circle [radius=0.1];
\draw (3, 2) circle [radius=0.1];

\draw (3, 0) circle [radius=0.1];
\draw (1, 0) circle [radius=0.1];
\draw (4, 1) circle [radius=0.1];
\end{tikzpicture}
\caption{Gadget $T_{q, n}$}
\label{injectclique}
\end{figure}

Let $\zG=(V_1,E_1,\pi_1),\zH=(V_2,E_2,\pi_2)$ be $k$-labeled 
graphs, and let $G=(V_1,E_1),H=(V_2,E_2)$.
Now we construct an instance $G',H'$ of
{\sc Graph Isomorphism}. Let $A_q=\{a_1, ..., a_q\}$. Then the graph $G'$ 
consists of a copy of $G$, a copy of $T_{q, n}$, and a copy 
of vertices from $A_q$ with the following additional edges: for each 
$i \in [q]$ the vertex $z_i$ from the $(i + 1)$th block of 
$T_{q, n}$ is adjacent to the vertex $a_i$. Also, we add edges 
from $G$ to vertices from $A_q$ : for a vertex $b\in V_1$ with 
$\pi_1(b)=i$, we add an edge $\{b, a_i\}$, see, Fig~\ref{G'}.

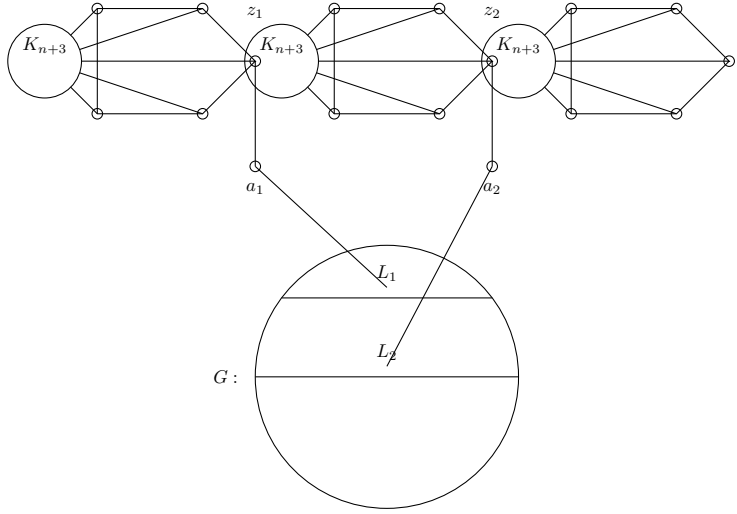
\begin{figure}[h]
\centering
\begin{tikzpicture}[scale=0.7, transform shape]
\draw (9, 1) --(10, 2) --(12, 2) --(13, 1) --(9, 1)-- 
               (10, 0) --(12, 0) --(13, 1);
\draw (9, 1) --(12, 2);
\draw (9, 1) --(12, 0);
\draw (10, 2) --(10, 0);
\draw (9, 1) circle [radius=0.1];

\draw (10, 2) circle [radius=0.1];
\draw (12, 2) circle [radius=0.1];
\draw (13, 1) circle [radius=0.1];
\draw (12, 0) circle [radius=0.1];
\draw (10, 0) circle [radius=0.1];
---------------
\draw [fill=white] (9, 1) circle [radius=0.7];
\node [above] at (9, 1) {$K_{n + 3}$};
\node [above] at (8.5, 1.7) {$z_2$};
----------------------------
\draw (4.5, 1) --(5.5, 2) --(7.5, 2) --(8.5, 1) --(4.5, 1)-- 
               (5.5, 0) --(7.5, 0) --(8.5, 1);
\draw (4.5, 1) --(7.5, 2);
\draw (4.5, 1) --(7.5, 0);
\draw (5.5, 2) --(5.5, 0);
\draw (4.5, 1) circle [radius=0.1];

\draw (5.5, 2) circle [radius=0.1];
\draw (7.5, 2) circle [radius=0.1];
\draw (8.5, 1) circle [radius=0.1];
\draw (7.5, 0) circle [radius=0.1];
\draw (5.5, 0) circle [radius=0.1];

\draw (8.5, -1) circle [radius=0.1];
\node [below] at (8.5, -1.2) {$a_2$};
\draw (8.5, 1) -- (8.5, -1);
---------------
\draw [fill=white] (4.5, 1) circle [radius=0.7];
\node [above] at (4.5, 1) {$K_{n + 3}$};
\node [above] at (4, 1.7) {$z_1$};

------------------------
\draw (0, 1) --(1, 2) --(3, 2)-- (4, 1) --(0, 1)-- 
               (1, 0) --(3, 0) -- (4, 1);
\draw (0, 1) --(3, 2);
\draw (0, 1) --(3, 0);
\draw (1, 2) --(1, 0);
\draw [fill=white] (0, 1) circle [radius=0.7];
\node [above] at (0, 1) {$K_{n + 3}$};

\draw (1, 2) circle [radius=0.1];
\draw (3, 2) circle [radius=0.1];

\draw (3, 0) circle [radius=0.1];
\draw (1, 0) circle [radius=0.1];
\draw (4, 1) circle [radius=0.1];

\draw (4, -1) circle [radius=0.1];
\node [below] at (4, -1.2) {$a_1$};
\draw (4, 1) -- (4, -1);

\draw (6.5, -5) circle [radius=2.5];
\draw (4, -5) --(9, -5);
\draw (4.5, -3.5) --(8.5, -3.5);
\node [above]  at (6.5, -3.3) {$L_1$};
\node [above]  at (6.5, -4.8) {$L_2$};
\node [left] at (3.8, -5) {$G:$};


\draw (4, -1) --(6.5, -3.3);
\draw (8.5, -1) --(6.5, -4.8);
\end{tikzpicture}
\caption{Constructing $G'$}

\label{G'}
\end{figure}

Graph $H'$ is constructed in a similar way. Let 
$T'_{q, n +3}$ be a copy of $T_{q, n}$ such that for each 
$i \in [q]$, $z_i$ is renamed to $z'_i$ in $T'_{q, n}$. Also, 
let  $A'_q$ be the set of vertices $\{a'_1, ..., a'_q\}$. Graph $H'$ 
consists 
of a copy of $H$, a copy of $T'_{q, n}$, and a copy of vertices 
from $A'_q$ with the following additional edges: for each 
$i \in [q]$, the vertex $z'_i$ from the $(i + 1)$th block of 
$T'_{q, n}$ is adjacent to the vertex $a'_i$. Also, we add 
edges from $H$ to vertices from $A'_q$: for a vertex $c \in V_2$ 
with $\pi_2(b)=i$, we add an edge $\{c, a'_i\}$. 

We need several properties of this construction. Some 
of them are proved in  \cite{Fomin15:lower} and the 
rest are proved here. 

\begin{lemma}\label{lem:claims}
\begin{itemize}
\item[(1)]
Any isomorphism $\vf$ from $G'$ to $H'$ maps $T_{q, n}$ to 
$T'_{q, n}$. 
\item[(2)]
Any isomorphism $\vf$ from $G'$ to $H'$ bijectively maps 
$T_{q, n}$ to $T'_{q, n}$ so that the for each 
$i \in [q]$, $z_i$ is mapped to $z'_i$.
\item[(3)]
Any isomorphism $\vf$ from $G'$ to $H'$ maps vertices from 
$A_q$ to vertices from $A'_q$ so that $a_i$ is mapped to 
$a'_i$ for each $i \in [q]$.  
\item[(4)]
Any isomorphism $\vf$ from $G'$ to $H'$ maps $G$ to $H$. 
\label{maps-G-to_H}
\end{itemize}
\end{lemma}

\begin{proof}
(1),(2) are proved in \cite{Fomin15:lower}.

\smallskip
(3) By items (1),(2)
the restriction 
of $\vf$ on $T_{q, n}$ is a bijective and surjective mapping 
from $T_{q, n}$ to $T'_{q, n}$, such that $\vf(z_i) = z'_i$ 
for each $i \in [q]$.  This also means that  each vertex of 
$T'_{q, n}$  is the image of some vertex from $T_{q, n}$ 
under $\vf$. Thus,  for each $i \in [q]$, $\vf(a_i)$ cannot be from 
$T'_{q, n}$ as otherwise $\vf$ would not be a bijection. On the 
other hand,  for each $i \in [q]$, since $\vf$ is an isomorphism and  
$(z_i, a_i)$ is an edge of $G'$, $\vf(a_i)$ is adjacent to 
$\vf(z_i) = z'_i$. Now, for each $i \in [q]$, the adjacent vertices 
of  $z'_i$ are either in $T'_{q, n}$  or it is $a'_i$. As $\vf(a_i)$ 
cannot be in $T'_{q, n}$ and has to be adjacent to $z'_i$, thus  
$\vf(a_i) = a'_i$. 

\smallskip

(4) By items (1)--(3) 
the restriction of $\vf$ on 
$T_{q, n}$ and $A_q$ is a bijective and surjective mapping 
from $T_{q, n}$ to $T'_{q, n}$, and from  $A_q$ 
to from $A'_q$. Thus each vertex from  
$T'_{q, n}$ or $A'_q$ is the image of some vertex in $G'$ 
that is also in  $T_{q, n}$ or $A_q$, under the $\vf$. Hence 
the image of each vertex of $G$ is a vertex from $H$, under $\vf$. 
\end{proof}

Now we can prove Proposition~\ref{part-iso-time}.

\begin{proof}[Proof of Proposition~\ref{part-iso-time}]
Let $G'$ and $H'$ be two graphs constructed from $G$ and $H$ 
as described above. First we show that there is a  
isomorphism from $\zG$ to $\zH$ if and only if there is an 
isomorphism from $G'$ to $H'$.

Let $\vf$ be an isomorphism from $\zG$ to $\zH$. We show that 
its natural extension $\vf'$ mapping $T_{q, n}$ to $T'_{q, n}$ 
and vertices from $A_q$ to vertices from $A'_q$ is an isomorphism 
from $G'$ to $H'$. The fact that it maps edges of $G'$ to edges of 
$H'$ is non-trivial only for edges of $G'$ from $G$ to vertices of 
$A_q$. Consider an edge from a vertex $b\in V_1$ with $\pi_1(b)=i$ 
to the vertex $a_i$. We have $\pi_2(\vf'(b))=i$ as the $\vf'$ is an 
extension of an isomorphism from $\zG$ to $\zH$. Also 
$\vf'(a_i)=a'_i$ means that $\vf'$ maps $(a_i, b)$ to an edge of $H'$. 

For the reverse direction, let $\vf'$ be an isomorphism from $G'$ 
to $H'$. We show that its restriction on $V(G)$ is an 
isomorphism from $\zG$ to $\zH$. Since by 
Lemma~\ref{lem:claims}(4) 
$\vf'$ maps $G$ to $H$, it is enough to check that $\vf'$ preserves
the labelling. In order to do this consider a vertex $b \in V_1$ with
$\pi_1(b)=i$. Then $\pi_2(\vf(b))=i$ as well, as otherwise 
$(\vf'(b), a'_i)$ would not be an edge of $H'$. 

Therefore, there is a polynomial time reduction the Isomorphism
problem for $k$-labeled graphs to {\sc Graph Isomorphism}
such that the number of vertices only grows by a constant 
factor $(k + 2)$. Now by \cite{Babai16:graph}
there is an algorithm that solves the instance $(G', H')$ in time 
$O(2^{log((k + 2)n )^{O(1)}})$.
The result follows.
\end{proof}

The following theorem is the main result of this section.

\begin{theorem}\label{the:k-expr-gen}
If graph $G$ has extended clique width $k$, then an extended 
$k$-expression for $G$ can be found in time 
$O^*((4k + 4)^{|V(G)|})$.
\end{theorem}

\begin{proof}
We create an array $N$ of size $(k + 1)^n$ whose entries $N(\zG')$ 
are labeled with a $k$-labelling $\zG'$ of a subgraph $G'$ of $G$. 
For any entry $N(\zG')$ the $k$-labeled graph $\zG'$ either has an 
extended $k$-expression or it does not. The goal is to set the value of 
each entry $N(\zG')$ to some extended 
$k$-expression for $\zG'$ if it has one and to "no" otherwise.   

Now we consider more detailed possibilities for each $\zG'$. 
There are four cases. Case~1 takes place if $\zG'$ has an 
extended $k$-expression that ends with a Beta operator; 
Case~2 takes place if $\zG'$ it 
has an extended $k$-expression that ends with a connect operator; 
Case~3 takes place if $\zG'$ has an extended 
$k$-expression that ends with a sequence of relabelling operators; 
and, finally, Case~4 takes place if $\zG'$ does not have an  
extended $k$-expression. 

All one-element $k$-labeled graphs are obviously represented by an 
extended $k$-expression. Let us suppose the values of each entry 
$N(\zG')$, where $\zG'$ contains at most $n-1$ vertices is  set 
correctly. Then, we want to set the correct values for entries of the 
array whose associated $k$-labeled graph has exactly $n$ vertices. 
We use the dynamic programming approach that consists of two 
phases. In Phase 1, for each entry $N(\zG')$ such that $\zG'$ has 
$n$ vertices, we check if $\zG'$ satisfies the conditions of Case~1. 
Then for each $k$-labeled graph like this that does not satisfy the
conditions of Case~1 we check if it falls in
Case~2. In Phase~2,  by relabelling $\zG'$ for which 
$N(\zG')$ is assigned a value, we find 
a new extended $k$-expression for $\zG'$ that do not satisfy 
the conditions of Cases~1 and~2, but satisfy the conditions of 
Case~3. In the end, for each 
$\zG'$ that belongs to none of Cases~1, 2, or~3, we 
set the value $N(\zG')$ to "no" because it does not have an extended 
$k$-expression. In the rest of this proof, for a $k$-labeled graph 
$\zG'$, we show how to check if it satisfies the conditions each of 
Cases~1 and~2.

Let $\zG'= (V', E', \pi')$ be a $k$-labeled graph with $|V'|=n$
and it has an extended $k$-expression that ends 
with a beta operator. Then there is an induced subgraph $\zG'_1$ 
of $\zG'$ such that the result of application of a beta operator to 
$\zG'_1$ is isomorphic to $\zG'$ and $\zG'_1$ has an extended 
$k$-expression. Thus, there exist  $\sigma: [k]\to[k]$, 
$\vec{n}\in (\{0\} \cup [k])^k$, $\cS \subseteq\cL(\vec{n})$, and 
a set $V'_1 \subset V'$, such that 
\begin{itemize}
\item[(A)] 
$\zG'_2=(V'_2,E'_2,\pi'_2)=\beta_{\vec{n},\zS,\sigma}(\zG'[V'_1])$ 
is isomorphic to $\zG$, and 
\item[(B)] 
$\zG'[V'_1]$ has an extended $k$-expression. 
\end{itemize}

Conversely, if there exist $V'_1\subset V'$, $\sigma: [k]\to[k]$, 
$\vec{n} \in (\{0\} \cup [k])^k$, $\cS\sse\cL(\vec{n})$ satisfying
conditions (A),(B), then $\zG'_2$ has an extended $k$-expression 
that ends with a beta  operator. As $\zG'_2$ and $\zG'$ are isomorphic, 
$\zG'$ has an extended $k$-expression that ends with a 
beta  operator as well.
Thus, the sufficient and necessary conditions for  $\zG'$ to have an 
extended $k$-expression  that ends with a Beta operator, is that 
there exist $V'_1\subset V', \sigma: [k]\to[k],
\vec{n}\in (\{0\}\cup [k])^k,\cS\in\cL({\vec{n}})$ satisfying 
(A),(B). 

The algorithm now searches through all possible
selection of $V'_1,\sigma,\cS$, to check if conditions (A),(B)
satisfied for any of them. Let us evaluate the running time of this 
procedure. Checking condition (A) takes time 
$O(2^{log((k + 2)n )^{O(1)}})$ by Proposition~\ref{part-iso-time}, 
while condition (B) can be verified by looking up the existing entry 
$N(\zG'[V'_1])$ in $O(1)$ time. There are $2^n$ choices for $V'_1$ 
and $k^k$ choices for $\sigma$. Vector $\vec{n}$ can be chosen in
$(k + 1)^k$ ways, and so $\cL(\vec{n})$ has at most 
$(k + 1)^2(k )^2$ elements. Thus, $\cS$ can be chosen in 
 at most $2^{(k + 1)^2(k )^2}$ ways. Thus, the total running time 
of filling up $N(\zG')$ in this case is upper bounded by 
\begin{align*}
2^{|V(G)|} \times k^k \times (k + 1)^2(k )^2 \times 
2^{(k + 1)^2(k )^2}\times O(2^{log((k + 2)n )^{O(1)}}) = 
O^*(2^{2|V(G)|}).
\end{align*}

Now let us suppose that $\zG' = (V', E', \pi')$ has an extended 
$k$-expression that ends with a connect operator. 
Then due to the safety of extended $k$-expressions there exist 
two induced subgraphs $\zG'_1$ and $\zG'_2$ of $\zG'$ such 
that, first, they both are represented by extended $k$-expressions, 
and, second, there is 
$\cT \subseteq [k]^2$, such that $\eta_{\cT}(\zG'_1, \zG'_2)$ is 
identical to $\zG'$. Thus to find an extended $k$-expression for 
$\zG'$ it suffices to go through all partitions of $V'$ into sets 
$V'_1$ and $V'_2$ and for each partition check the following two 
conditions. First, check if $\zG'_1 = \zG'[V'_1]$ and 
$\zG'_2 = \zG'[V'_2]$ have an extended $k$-expression by 
looking up the entries $N(\zG'_1),N(\zG'_2)$. Second, check if 
there is $\cT \subseteq [k]^2$ such that $\eta_{\cT}(\zG'_1, \zG'_2)$ 
is identical to $\zG$. Since there are at most $2^{|V(G)|}$ 
ways to partition $V'$ into $V'_1$ and $V'_2$, takes time
$O(2^{|V(G)|})$ to check if $\zG'$ falls into Case~2.

So far we have registered an extended $k$-expression for every $\zG'$ 
that satisfies the conditions of Case~1 or Case~2. Now, start 
Phase~2 and check 
whether any of the remaining $k$-labeled graphs $\zG'$ satisfies 
the conditions of Case~3. In order to do that we go through all 
$k$-labeled graphs $\zG'$ with $n$ vertices and such that 
$N(\zG')$ contains an extended $k$-expression $\Phi$, that is 
initially for all $\zG'$ that fall into Cases~1,2. Then we 
consider every possible single relabelling $\rho_{ij}$ in turn. 
If $\rho_{ij}(\zG')$ is a $k$-labeled graph such that $N(\zG')$ 
does not have an extended $k$-expression, then we set 
$N(\rho_{ij}(\zG'))=\rho_{ij}(\Phi)$.  
We repeat this process for each $k$-labeled graph $\zG'$, 
until no new entries can be filled. The time required for 
Phase 2 in total, for all $\zG'$, not only those with $n$ vertices is 
bounded by number of all $k$-labellings of all subgraphs of $G$ 
times the number of possible operators $\rho_{ij}$. As is easily seen, 
the time required for Phase~2 in total is 
\begin{align*}
(k + 1)^{|V(G)|} \times k^2 = O^*((k + 1)^{|V(G)|})
\end{align*}

\textit{Time complexity}: The array we construct has 
$(k + 1)^{|V(G)|}$ entries. The time required to 
complete Phase~1 for all the entries is bounded 
by $O^*(4^{|V(G)|} \times (k + 1)^{|V(G)|} )$. The time to 
complete Phase 2 for all entries is bounded by 
$O^*((k + 1)^{|V(G)|})$. Thus the total running time is 
$O^*((4k + 4)^{|V(G)|})$.
\end{proof}

\section{Counting homomorphism to labeled graphs given an
extended $k$-expression}\label{sec:count}

In this section we prove our main result. Let $\hom(G,H)$
denote the number of homomorphisms from a graph $G$ to 
a graph $H$.

\begin{theorem}\label{main-theorem}
Let $G$ and $H$ be two graphs, and let $k$-labeled graph 
$\zH$ be a \textit{$k$-labelling} of graph $H$. Given an extended 
$k$-expression $\Phi$ for $\zH$, $\hom(G, H)$ can be found in time 
$O^*((2k + 1)^{|V(G)|})$
\end{theorem}

The following notation and terminology will be used throughout 
this section.
The number of homomorphism from graph $G$ to graph $H$ is 
denoted by $\hom(G, H)$. Also, $\HOM(G, H)$ denotes the set of all 
homomorphisms from $G$ to $H$.
Let $X\subseteq V(G)$, and let $\chi:X\to[k]$ be a label function. 
A mapping $\vf$ from $X$ to $k$-labeled graph $\zH = (V, E, \pi)$ is said 
to be consistent with $\chi$ if for every $x \in X$ it holds 
$\pi(\vf(x)) = \chi(x)$. 
Let $\hom_{\chi}(G, \zH)$, $\HOM_{\chi}(G, \zH)$, 
$\map_{\chi}(G, \zH)$, and $\MAP_{\chi}(G, \zH)$, denote the 
number of homomorphisms from $G[X]$ to $\zH$ consistent with 
$\chi$, the set of all homomorphisms from $G[X]$ to $\zH$ consistent 
with $\chi$, the number of all mappings from $G[X]$ to $\zH$ 
consistent with $\chi$, and the set of all mappings from $G[X]$ to 
$\zH$ consistent with $\chi$, respectively.

Let $\Phi$ be an extended $k$-expression for a $k$-labelling 
$\zH$ of the graph $H$. We proceed by induction on the structure 
of $\Phi$. More precisely, our algorithm will compute entries
$\hom(\zG[X],\zH')$, where $X\sse V(G)$ and $\zG$ is a $k$-labelling
of $G$, and $\zH'$ is the $k$-labeled graph represented by a 
subformula of $\Phi$. Operator $\cdot_i$ creating a graph $\zH'$
with a single vertex labeled $i$ gives the base case of induction. 
In this case $\hom(\zG[X],\zH')=1$ if all vertices of $X$ are 
labeled $i$ and $\zG[X]$ has no edges; otherwise 
$\hom(\zG[X],\zH')=0$. Finally, after computing the numbers
$\hom(\zG,\zH)$ for all the $k$-labellings $\zG$ of $G$ we complete 
using the following observation.

\begin{observation}\label{exhaustive-observation}
Let $G$ and $H$ be graphs, and let $k$-labeled graph $\zH $ be a 
$k$-labelling of $H$.
Then 
\begin{align*}
\hom(G, H) = \sum_{\chi:V(G)\to[k]}\hom_{\chi}(G, \zH)
\end{align*}
\end{observation}

It therefore suffices to show how to compute $\hom(\zG[X],\zH')$, 
where $\zG$ is an arbitrary $k$-labelling of $G$, $X\sse V(G)$, and
$\zH'$ is represented by a subformula $\Phi'$ of $\Phi$, provided 
$\hom(\zG[Y],\zH'')$ is known for all $Y\sse X$ and $\zH''$ 
represented by a subformula $\Phi''$ of $\Phi'$ with $\Phi''\ne\Phi'$.
We consider 3 cases depending on the last operator of $\Phi'$.

\subsection{Relabelling Operator}

Let $\Psi$ be a string of relabelling operators, applied on a 
$k$-labeled graph $\zG$. We say that $\Psi$ relabels $i$ to $j$, 
if the application of $\Psi$ on a vertex labeled with $i$ gives a 
vertex labeled with $j$. 

Let $\zH' = ( V', E', \pi')$ be a $k$-labeled graph, and let 
$\zH=\Psi(\zH')=(V,E,\pi)$ be the result of application of 
$\Psi$ to $\zH'$. Let $X \subseteq V(G)$ and let 
$\chi: X \to [k]$ be any $k$-label function. Also, let 
$\vf \in \HOM_{\chi}(G, \zH)$. Then there is a function 
$\chi': X \to [k]$ such that $\vf\in \HOM_{\chi'}(G, \zH')$. 
Let call $\chi'$ the \textit{consistent labelling} of $\vf$. To count
the elements of $\HOM_{\chi}(G, \zH)$, we partition this set 
into sets of homomorphisms that share the same consistent labelling, 
and find the number of elements of these smaller sets.

Let $\cD(\chi)$ denote the set of all of functions 
$\chi': X \rightarrow [k]$ that satisfy  $\chi = \sigma(\chi') $.
We show that $\cD(\chi)$  is the set of consistent labellings of all 
homomorphisms from $\HOM_{\chi}(G, \zH')$. 

\begin{lemma}\label{relabel_lemma}
Let $\chi$, $G$,  $\zH$, and $\zH'$  be as above. Then,
\begin{align}
    \hom_{\chi}(G, \zH) = \sum_{\chi' \in \cD(\chi)} 
|\HOM_{\chi'}(G, \zH')|.
\end{align}
\end{lemma}

\begin{proof}
Let $\vf\in \HOM_{\chi}(G, \zH)$. Also, let $\chi': X \rightarrow [k]$, 
be the consistent labelling of $\vf$. As is easily seen, $\chi'$ 
satisfies $\chi = \sigma(\chi')$. Now, let $\chi' \in \cD(\chi)$, then 
the set $\HOM_{\chi'}(G, \zH')$ is the set of all elements of 
$\HOM_{\chi}(G, \zH)$ that have $\chi'$ as their consistent labelling.  
Note that for two different $\chi', \chi'' \in \cD(\chi)$ the sets 
$\HOM_{\chi'}(G, \zH')$ and $\HOM_{\chi''}(G, \zH')$ are disjoint.  
Thus, 
\begin{align*}
   \HOM_{\chi}(G, \zH) = \bigcup_{\chi' \in \cD(\chi)} 
\HOM_{\chi'}(G, \zH') 
\end{align*}
The result follows.
\end{proof}

We now have an algorithm computing all the required numbers of 
homomorphisms.

\begin{corollary}
Let $Y\sse V(G)$ and $\gm\colon Y \to [k]$. There is an algorithm 
that given $\hom_{\zeta}(G, \zH')$  for all 
functions $\zeta$ from subsets of $V(G)$ to $[k]$ as input, 
finds $\hom_{\gm}(G, \zH)$ in time $O^*(2^{|V(G)|})$. 
\end{corollary}

\begin{proof}
First for a fixed $X \in V(G)$ and a fixed function $\chi: X \to [k]$ 
we calculate the upper bound on the size of $\cD(\chi)$. For a 
function $\chi' \in \mathbb{D}$, we can count those functions 
by considering their possible values on each element of $X$. 
Since the codomain of $\chi'$ is $[k]$ and size of $X$ is bounded 
by $|V(G)|$, there are at most $k^{|V(G)|}$  possibilities for 
$\chi'$. By Lemma~\ref{lemma-join-operation} 
$\hom_{\gm}(G, \zH)$ can be found by summing over 
$O^{*}(k^{|V(G)|})$ known values. That can be done in 
time $O^{*}(k^{|V(G)|})$.
\end{proof}

\subsection{Connect Operators} \label{section-connect}

Let $\zH' = (V', E', \pi')$ and $\zH'' = (V'', E'', \pi'')$ be $k$-labeled 
graphs, let $\cT$ be a subset of $[k] \times [k]$, and let 
$\zH = \eta_{\cT}(\zH', \zH'')$ be a $k$-labeled graph. Let 
$X\sse V(G)$, and let $\chi:X\to[k]$ be any $k$-label function.

To count elements of $\HOM_{\chi}(G,\zH)$, one can partition them 
into smaller subsets of elements and count the elements in each 
subset. For any $\vf\in\HOM_{\chi}(G,\zH)$, let 
$\vf' = \vf\vert_{X'}$ and $\vf'' = \vf\vert_{X''}$, where $X'$ and $X''$ 
are preimages of $V(\zH')$ and $V(\zH'')$ under $\vf$, respectively. 
Then $\vf'$ and $\vf''$ are homomorphisms from $G[X']$ to $\zH'$ and 
from $G[X'']$ to $\zH''$, respectively. Thus, there are unique 
$k$-label functions 
$\chi':X' \to V(\zH')$ and $\chi'':X'' \to V(\zH'')$ with which $\vf'$ 
and $\vf''$ are consistent. Let us call the pair of $k$-label 
functions $(\chi', \chi'')$, the \textit{consistent pair} of $\vf$. The 
idea is to partition $\HOM_{\chi}(G, \zH)$ into smaller subsets 
such that for each of them all the homomorphisms $\vf$ in 
that subset share the same consistent pair. 

Let $\All(\chi)$ be the set of all pairs of $k$-label functions 
$(\chi',\chi'')$, $\chi':X' \to V(\zH'), \chi'':X'' \to V(\zH'')$, where 
$X'$ and $X''$ are disjoint subsets of $X$ and: 
\begin{itemize}\setlength{\itemindent}{.5cm}
\item[(a.1)] 
$X'\cup X''=X$; and also $\chi\vert_{X'}=\chi'$ and 
$\chi\vert_{X''}=\chi''$.  
\item[(a.2)] 
For any $x' \in X'$, such that $\chi'(x) = i $ and for any $x'' \in X''$ 
such that $\chi''(x'') = j$, if $x'x''\in E(G)$ then, $(i, j) \in \cT$.
\end{itemize}

We show that $\All(\chi)$ is the set of consistent pairs of 
homomorphisms from $\HOM_{\chi}(G, \zH)$. Also for every pair 
$(\chi', \chi'') \in\All(\chi)$ we count the number of 
$\vf\in\HOM_{\chi}(G, \zH)$, such that $(\chi', \chi'')$ is 
the consistent pair of $\vf$.  

\begin{lemma}\label{lemma-join-operation}
The number $\hom_{\chi}( G, \zH)$ satisfies the following equality 
\begin{align}
\hom_{\chi}(G, \zH) = \sum_{(\chi', \chi'') \in\All(\chi) }
\hom_{\chi'}( G, \zH')\hom_{\chi''}( G, \zH'').
\label{recursion}
\end{align}
\end{lemma}

\begin{proof} 
First, observe that the right hand side of (\ref{recursion}) 
counts the elements of the set
\[
R = \bigcup_{(\chi', \chi'') \in\All(\chi)}\HOM_{\chi'}(G, \zH') 
\times\HOM_{\chi''}(G, \zH'').
\]
Also, the left hand side equals the cardinality of the set 
$\HOM_{\chi}(G, \zH)$ by definition. So, (\ref{recursion}) can be 
proved by constructing a bijection between 
$\HOM_{\chi}(G, \zH)$ and $R$.

Let $\vf\in\HOM_{\chi}(G, \zH)$, and let $X'$ and $X''$ be preimages 
of $V(\zH')$ and $V(\zH'')$ under $\vf$, respectively. Consider the 
mapping $\tau:\vf\mapsto(\vf \vert_{X'}, \vf \vert_{X''})$.

To show that $\tau$ is a bijection between $\HOM_{\chi}(G, \zH)$ 
and $R$, one needs to show that, first, for any 
$\vf\in\HOM_{\chi}(G, \zH)$, $\tau(\vf)$ is an element of $R$; and, 
second, that $\tau$ is surjective and injective.

Let $\vf' = \vf\vert_{X'}$ and $\vf'' = \vf\vert_{X''}$. As is easily seen, 
$\vf'$ and $\vf''$ are elements of $\HOM(G, \zH')$ and 
$\HOM(G, \zH'')$, respectively. Also, there is a unique pair of 
functions $\chi':X' \to [k]$ and $\chi'':X'' \to [k]$ such that $\vf'$ 
and $\vf''$ are consistent with $\chi'$ and $\chi''$, respectively, i.e., 
$\vf' \in\HOM_{\chi'}(G, \zH')$ and $\vf'' \in\HOM_{\chi''}(G, \zH'')$. 
Since $\vf' = \vf\vert_{X'}$ and $\vf'' = \vf\vert_{X''}$ are consistent 
with $\chi'$ and $\chi''$, respectively,  
$\chi \vert_{X'} = \chi'$ and $\chi \vert_{X''} = \chi''$. This means 
that  $(\chi', \chi'')$ satisfies Property~(a.1).

Let $x'\in X'$ and $x''\in X''$ be such that $\chi'(x') = i$ and 
$\chi''(x'') = j$ and also, $x'x''\in E(G)$. Then, 
$\chi(x') = \chi'(x') = i$ and $\chi(x'') = \chi(x'') = j$.  Since 
$\vf(x')\vf(x'')$ is an edge of $\zH = \eta_{\cT}(\zH', \zH'')$, 
then $(i, j) \in\cT$. Thus $(\chi', \chi'')$ satisfies Properties~(a.1) 
and~(a.2), therefore $(\chi', \chi'') \in\All(\chi)$. This means 
$(\vf', \vf'') \in\HOM_{\chi'}(G, \zH') \times\HOM_{\chi''}( G, \zH'') 
\subseteq R$. In other 
words, we have shown that for any $\vf\in\HOM_{\chi}(G, \zH)$, 
the consistent pair of $\vf$ belongs to $\All(\chi)$.

For two nonidentical elements $\vf$ and $\psi$ of 
$\HOM_{\chi}(G, \zH)$, $\tau(\vf)$ is clearly not equal to 
$\tau(\psi)$, so $\vf$ is injective.

Finally, we show that $\tau$ is surjective. Let $\chi':X' \to [k]$ and 
$\chi'':X'' \to [k]$ be any two functions such that 
$(\chi', \chi'')\in\All(\chi)$. Let $\vf'$ and $\vf''$ be two 
homomorphisms in $\HOM_{\chi'}(G, \zH')$ and 
$\HOM_{\chi''}(G, \zH'')$, respectively, i.e., $(\vf', \vf'') \in R$. 
Let
\[
\vf(x) = 
\begin{cases}
\vf'(x), &\text{$x \in X'$,}\\
\vf''(x), &\text{$x \in X''$}.
\end{cases}
\]
By Property~(a.1), $\chi \vert_{X'} = \chi'$ and 
$\chi \vert_{X''} = \chi''$. Then, $\vf$ is consistent with $\chi$, i.e. 
$\vf\in\MAP_{\chi}(G, \zH)$. 

Let $x' \in X'$ and $x''\in X''$ be such that $x'x''\in E(G)$. To 
show that $\vf$ is a homomorphism, observe that $\vf(x') = \vf'(x')$ 
and $\vf(x'') = \vf''(x'')$. 
Thus, $\pi'(\vf'(x') ) = \chi'(x') = i$ and 
$\pi''(\vf''(x'')) = \chi''(x'') = j$. 
Since $x'x''\in E(G)$ by Property~(a.2), $(i, j)\in\cT$ and thus 
$\vf'(x')\vf''(x'')) \in E(\zH)$. Therefore, $\vf$ is a homomorphism.

For an arbitrary element $(\chi', \chi'')$ of $\All(\chi)$ and an 
arbitrary element $(\vf', \vf'')$ of $\HOM_{\chi'}(G, \zH') \times 
\HOM_{\chi''}(G, \zH'') \subseteq R$, there is $\vf$ from 
$\HOM_{\chi}(G, \zH)$ such that $\tau(\vf) = (\vf',\vf'')$ proving 
$\tau$ is surjective. 

The bijection $\tau$ between $\HOM_{\chi}(G, \zH) $ and $R$, 
proves (\ref{recursion}). 
\end{proof}

\begin{lemma}
Let $X$ be a subset of $V(G)$ and let $\chi$ be a function 
$\chi\colon X \to [k]$. There is an algorithm that given 
$\hom_{\zeta}(G, \zH')$ and $\hom_{\zeta}(G, \zH'')$ for all 
functions $\zeta$ from a subset of $V(G)$ to $[k]$ as input, 
finds $\hom_{\chi}(G, \zH)$ in time $O^*(2^{|V(G)|})$. 
\end{lemma}

\begin{proof}
By Lemma~\ref{lemma-join-operation} $\hom_{\chi}(G, \zH)$ 
can be found by summing over $O^{*}(2^{|V(G)|})$ known 
values. 
\end{proof}

\subsection{Beta operators} \label{section-beta-operator}

In this part, we show how to make a recursive step in the case 
when the last operator of $\Phi'$ is a beta operator. 
Before explaining this step, we need several definitions.

A \emph{retraction} is a homomorphism $\psi$ from a graph $G_2$ 
to its subgraph $G_1$ such that $\psi(v) = v$ for each vertex $v$ of 
$G_1$. In this case the subgraph $G_1$ is called a 
\emph{retract} of $G_2$. A retraction from a $k$-labeled graph 
$\zG_2 = (V_2,E_2, \pi_2)$ to a $k$-labeled graph 
$\zG_1 = (V_1, E_1, \pi_1)$ is defined to be a retraction from 
$G_2 = (V_2, E_2)$ to $G_1 = (V_1,E_1)$ preserving 
the label function $\pi_2$, that is, $\pi_2(v) = \pi_1(\psi(v))$ 
for all $v \in V_2$.

It will be convenient for us to subdivide operator 
$\beta_{\vec{n}, \sigma, \cS}$ into two steps: the first one is 
expansion of the original graph using $\vec{n}$ and $\cS$, and the 
second is relabelling of some vertices of the resulting graph 
using $\sigma$. More specifically, let $\zH=(V,E,\pi)$ be a $k$-labeled 
graph, $\vec{n}\sse([k]\cup\{0\})^k$, $\cS\sse\cL(\vec{n})$, 
and $\sigma:[k]\to[k]$. Then 
$\zH' =(V',E',\pi')=\alpha_{\vec{n}, \cS}(\zH)$ is given by 
\begin{itemize}
\item 
$V' = \bigcup_{i = 1}^{k} C_i$, where 
$C_i = \{a_{j}| j \in \{0, ..., n_i\},\text{ $a\in V$ 
and } \pi_1(a) = i\}$. The vertices $a_0$, $a\in V$, are called 
original vertices of $\zH'=\alpha_{\vec{n}, \cS}(\zH)$ and are 
identified with their corresponding vertices from $V$;
\item 
$(a_{j}, b_{j'})\in E'$ if and only if
$(a, b)\in E$, and $(\pi(a), j,\pi(b) ,j') \in\cS$ or  
$j = j' = 0$;
\item 
$\pi'(a_j)=\pi(a)$.
\end{itemize}
Then $\zH''=(V'',E'',\pi'')= \beta_{\vec{n}, \sigma, \cS}(\zH)$,
that is, $V''=V'$, $E''=E'$, and 
\[
\pi''(a_{j}) = \begin{cases}
\pi(a), & \textit{if $j = 0$,}\\
\sigma(\pi(a)), &\text{otherwise.}
\end{cases}. 
\]
As is easily seen, $\zH$ is an induced subgraph of $\zH'$, and a 
retract. Indeed, the mapping $\mu$ that 
maps every $o_j \in V(\zH')$ to $o$ (recall that $o_j$ is a 'copy' 
of some $o \in V(\zH)$) is a retraction.  

The objective is to find a method to express 
the number of homomorphisms from induced subgraphs of $G$ to 
$\zH''$ given those from induced subgraphs of $G$ to $\zH$. 

\begin{lemma}\label{lemma-beta-operation}
Let $Y\sse V(G)$ and let $\gm$ be a function $Y \to [k]$. 
There is an algorithm that given $\hom_{\zeta}(G, \zH)$ for all 
functions $\zeta$ from a subset of $V(G)$ to $[k]$ as input, finds 
$\hom_{\gm}(G, \zH'')$ in time $O((2(k+1))^{2|V(G)|})$. 
\end{lemma}

We break this down into two steps. 
The main result of Step~I, which is summarized in 
Lemma~\ref{lemma-H'-H-sum}, finds an equality for the number 
of homomorphisms from $G$ to $\zH'$. Then, the main result of 
Step~II, Lemma~\ref{lemma-H''-H'-sum}, 
finds the number of homomorphisms from induced subgraphs 
of $G$ to $\zH''$ given those for $G$ and $\zH'$.

\subsubsection*{Step 1}

Let $\zH' = (V', E', \pi')=\alpha_{\vec{n}, \cS}(\zH)$ and 
$\zH = (V, E, \pi)$. Let  $ Y$ be a subset of  $V(G)$ and 
$\gm$ a function $Y \to [k]$. Also, set
\[
\cW(\gm) = \{\omega| \text{ $\omega:Y \to [k] \cup \{0\}$ 
and $\forall a \in Y$, } 0 \leq  \omega(a) \leq n_{\gm(a)}  \}.
\]
For $\omega \in \cW(\gm)$, let 
\begin{align*}
\HOM_{\gm}(G, \zH', \omega) = &
\{\vf\mid \text{ $\vf \in\HOM_{\gm}(G, \zH')$ and $\forall a\in Y, 
\exists o \in V(\zH)$}\\
& \quad\text{such that $\vf(a) = o_{\omega(a)}$}\},
\end{align*}
and 
\begin{align*}
\MAP_{\gm}(G, \zH', \omega) = & \{\vf\mid 
\text{ $\vf \in\MAP_{\gm}(G, \zH')$ and $\forall a\in Y, 
\exists o \in V(\zH)$}\\
& \quad\text{ such that $\vf(a) = o_{\omega(a)}$}\}.
\end{align*}

For the rest of Step~I, let $X'$ and $X''$ be two disjoint subsets of 
$V(G)$ and let $\chi':X' \to [k]$ and $\chi'':X'' \to [k]$ be arbitrary 
functions. Also  let $X = X' \cup X''$ and let $\chi=\chi'\uplus\chi''$.

Let $\HOM_{\chi', \chi''}(G, \zH')$ denote the set of all 
elements of $\HOM_{\chi}(G, \zH')$ that map a vertex $a$ from $V(G)$ 
to an original vertex of $\zH'$ (recall that any vertex of $\zH$ is 
called an original vertex of $\zH'$) if and only if $a\in X'$. The set 
$\HOM_{\chi', \chi''}(G, \zH') $ 
is an intermediate set that is computed in Step~I. 
In Step~2, this set is used to calculate another value. 

For any $\vf\in\HOM_{\chi', \chi''}(G, \zH')$, there is a unique 
$\omega \in \cW(\chi)$ such that $\vf$ is also an element of 
$\HOM_{\chi}(G, \zH', \omega)$. Let us call $\omega$ the 
\textit{consistent function} of $\vf$. Similar to the method used in 
Section~\ref{section-connect}, one can partition 
$\HOM_{\chi', \chi''}(G, \zH')$ into smaller subsets and count the 
elements in each smaller subsets. We partition 
$\HOM_{\chi', \chi''}(G, \zH')$ into sets of homomorphisms 
that all share the same consistent function $\omega$  for some 
$\omega \in \cW(\chi)$. As is easily seen, 
$\HOM_{\chi}(G, \zH', \omega) \cap\HOM_{\chi', \chi''}(G, \zH')$ 
is such a subset.

Let $\cB(\chi', \chi'')$ be the set of all $\omega \in \cW(\chi)$ 
such that $\omega$ satisfies the following properties: 
\begin{itemize}\setlength{\itemindent}{.5cm}
\item[(b.1)] 
$\omega \in \cW(\chi)$ and  $\omega(x) = 0$ iff $x \in X'$. 
\item[(b.2)] 
For every $a, b \in X$ such that at least one of them is not an 
element of $X'$, and $ab\in E(G)$ it holds that 
$(\chi(a), \omega(a), \chi(b), \omega(b)) \in \cS$. 
\end{itemize}
We show that, the number of elements in $HOM_{\chi}(G, \zH')$ 
such that $\omega$ is their consistent function is the same  for any 
$\omega\in\cB(\chi', \chi'')$ and it is zero otherwise. 

\begin{lemma}\label{lemma-H'-H-sum}
Let $G$, $\zH$, $\zH'$, $\chi'$, and $\chi''$ be defined as above, 
then 
\[
|\HOM_{\chi', \chi''}( G,\zH')| = |\cB(\chi', \chi'')| \times
\hom_{ \chi}( G, \zH).
\]
\end{lemma}

\begin{proof}
One can easily observe that if $\omega \in \cW(\chi)$ does not 
satisfy Property~(b.1), then the number of 
$\vf\in\HOM_{\chi', \chi''}(G, \zH')$ that have $\omega$ as their 
consistent function is zero. Otherwise, if $\omega \in \cW(\chi)$ 
satisfies Property~(b.1), then $\HOM_{\chi}(G, \zH', \omega)$ is 
the set of all $\vf\in\HOM_{\chi', \chi''}(G, \zH)$ that have 
$\omega$ as their consistent function. So, 
\begin{align*}
\HOM_{X', X''}(G, \zH') = \bigcup_{\substack{\omega \in \cW(\chi)\\ 
\text{$\omega(x) = 0$ iff $x \in X'$}}}
\HOM_{\chi}(G, \zH', \omega).
\end{align*}

Based on this, to prove the Lemma~\ref{lemma-H'-H-sum}, it suffices 
to prove that if $\omega \in \cW(\chi)$ and it satisfies Property~(b.1) 
but it does not satisfy Property~(b.2), then the size of 
$\HOM_{\chi}(G, \zH', \omega)$ is zero, and otherwise, if  
$\omega \in \cW(\chi)$ satisfies Properties~(b.1) and~(b.2), then 
the size of $\HOM_{\chi}(G, \zH', \omega)$ is equal to the size of 
$\HOM_{\chi}(G, \zH)$.

Let $\omega$ be a function that satisfies Property~(b.1) but 
$\omega$ does not satisfy Property~(b.2), i.e., $\omega$ is not 
an element of $\cB(\chi', \chi'')$, and let $\vf$ be an element of 
$\HOM_{\chi}(G, \zH', \omega)$. Since $\omega$  does not satisfy 
Property~(b.2), there is $a, b \in \zX$ such that at least one of them 
is not in $X'$ and $ab\in E(G)$ and 
$(\chi(a),\omega(a),\chi(b),\omega(b)) = (\pi'(\vf(a)),\omega(a), 
\pi'(\vf(b)), \omega(b) ) \not \in \cS$. So, $(\vf(a),\vf(b))$ is not an 
edge in $\zH'$ and $\vf$ is not a homomorphism. Thus 
$|\HOM_{\chi}(G, \zH', \omega)| = 0$.

Now, for a fixed $\omega \in \cW(\chi)$ such that $\omega$ 
satisfies Properties~(b.1) and~(b.2), i.e., 
$\omega \in \cB(\chi', \chi'')$, one can define a bijection between 
the set $\HOM_{\chi}(G, \zH', \omega)$ and the set 
$\HOM_{\chi}(G, \zH)$, as follows. Suppose that  $\vf$ is an element 
of $HOM_{\chi}(G, \zH)$. The mapping 
$\tau_\vf: a \mapsto \vf(a)_{\omega(a)}$ is an element of 
$\MAP_{\chi}(G, \zH, \omega)$. We claim that mapping 
$\tau: \vf \mapsto \tau_{\vf}$ is a bijection between 
$\HOM_{\chi}(G, \zH)$ and $\HOM_{\chi}(G, \zH', \omega)$. 
In so doing, one should prove that first: $\tau_\vf$ is in 
$\HOM_{\chi}(G, \zH', \omega)$; and second, 
$\tau:\HOM_{\chi}(G, \zH) \to\HOM_{\chi}(G, \zH', \omega)$ is 
injective and surjective.

For any $a,b\in X'$ and $ab\in E(G)$, as $\omega$ satisfies 
Property~(b.1), $\omega(a) = \omega(b) = 0$. Thus, 
$\tau_\vf(a) = \vf(a)_{\omega(a)} = \vf(a)_0 = \vf(a)$ and 
$\tau_\vf(b) = \vf(b)_{\omega(b)}= \vf(b)_0 = \vf(b)$. 
Since $\vf$ is a homomorphism, $\tau_\vf(a)\tau_\vf(b)=  
\vf(a)\vf(b)\in E(\zH) \subseteq E(\zH')$. Thus the image 
of $ab$ under $\tau_\vf$ is an edge.

For any $a,b \in X$, such that at least one of $a$ and $b$ is not an 
element of $X'$, and $ab\in E(G)$, by definition 
$\tau_\vf(a)=\vf(a)_{\omega(a)}$ and 
$\tau_\vf(b)=\vf(b)_{\omega(b)}$.
Moreover, since $\pi(\vf(a)) = \chi(a)$ and $\pi(\vf(b)) = \chi(b)$,  
as $\omega$ satisfies Property~(b.2),  $(\chi(a), \omega(a), \chi(b), 
\omega(b) ) = (\pi(\vf(a)), \omega(a), \pi(\vf(b)), \omega(b))\in\cS$. 
Consequently, by the definition of operator 
$\alpha_{\vec{n},\cS}$, as $\vf(a)\vf(b)\in E(\zH)$, then 
$\vf(a)_{\omega(a)}\vf(b)_{\omega(b)}$ is an edge of $\zH'$. 
Thus for any $\vf\in\HOM_{\chi}(G, \zH)$, $\tau_\vf$ is a 
homomorphism. Since $\tau_\vf\in\MAP_{\chi}(G, \zH', \omega)$ for 
every $\vf\in\HOM_{\chi}(G, \zH)$,  and $\tau_\vf$ is a 
homomorphism, $\tau_\vf$ is an element of the set 
$\HOM_{\chi}(G, \zH', \omega)$.

Second step is to prove that $\tau$ is injective. Consider $\vf\neq\psi$ 
two different elements of $\HOM_{\chi}(G, \zH)$. There is an 
element $a$ in their domain such that $\vf(a) \neq\psi(a)$. Since 
$\tau_\vf(a) = \vf(a)_{\omega(a)} \neq \tau_\psi(a) = 
\psi(a)_{\omega(a)}$ we obtain $\tau_\vf\neq\tau_\psi$. 

Finally, we should prove that $\tau$ is surjective. For any 
homomorphism $\vf'\in\HOM_{\chi}(G, \zH', \omega)$, 
$\mu(\vf')$ is an element of $\HOM_{\chi}(G, \zH)$ (recall that $\mu$ 
is the natural retraction from $\zH'$ to $\zH$.). As is easy to 
see, $\tau_{\mu(\vf')} = \vf'$. Thus $\tau$ is surjective. 

The bijection $\vf$ between $\HOM_{\chi}(G, \zH)$ and 
$\HOM_{\chi}(G, \zH', \omega)$ proves that the cardinalities 
of these two sets are equal. 
\end{proof}

\subsubsection*{Step II}

Recall that $\zH'' = \beta_{\vec{n}, \sigma, \cS}(\zH)$ and 
$\zH' = \alpha_{\vec{n}, \zS}(\zH)$. Let $X$ be a subset of $V(G)$ 
and let $\chi$ be a $k$-label function from $X$ to $[k]$. Also, let 
$\vf\in\HOM_{\chi}(G, \zH'')$, and let $X'$ and $X''$ be the 
preimages of $V(\zH)$ and $V(\zH'') \setminus V(\zH)$ under $\vf$, 
respectively. Then, there is a pair of functions $\chi':X' \to [k]$ and 
$\chi'':X'' \to [k]$ such that $\vf\in\HOM_{\chi', \chi''}(G, \zH')$. 
Let us call $(\chi', \chi'')$ the \textit{consistent partition} of $\vf$. 
To count the elements of the set $\HOM_{\chi}(G, \zH'')$, we 
partition this set into smaller sets of homomorphisms that share the 
same  consistent partition, and count the number of elements in the 
smaller sets.

Let $\cC(\chi)$ denote the set of all pairs of functions 
$(\chi', \chi'')$, $\chi':X' \to [k], \chi'':X'' \to [k]$, $X',X''$ disjoint, 
that satisfy the following conditions:
\begin{itemize}\setlength{\itemindent}{.5cm}
\item[(c.1)] 
$X' \cup X'' =X$, and
\item[(c.2)] 
$\chi \vert_{X'} = \chi'$ and $\chi \vert_{X''} = \sigma(\chi'')$.  
\end{itemize}

We show that $\cC(\chi)$ is the set of consistent partitions of the 
homomorphisms in $\HOM_{\chi}(G, \zH'')$. 

\begin{lemma}\label{lemma-H''-H'-sum}
Let $\chi$, $G$, $\zH'$, and $\zH''$ be defined as above. Then, 
\begin{align}
\hom_{\chi}( G, \zH'') = \sum\limits_{(\chi', \chi'') \in \cC(\chi)} 
|\HOM_{\chi', \chi''}( G, \zH')|.
\label{recoloring-alpha-product}
\end{align}
\end{lemma}

\begin{proof}
Let $\vf\in\HOM_{\chi}(G, \zH'')$, and let $X'$ and $X''$ be the 
preimages of $V(\zH)$ and $V(\zH'') \setminus V(\zH)$ under $\vf$, 
respectively. Also let $(\chi',\chi'')$, 
$\chi: X' \to [k], \chi'': X'' \to [k]$, be the consistent partition of 
$\vf$. As is easily seen, $(\chi', \chi'')$ satisfies 
Properties~(c.1) and~(c.2).

Now, let $(\chi', \chi'') \in \cC(\chi)$, then the set 
$\HOM_{\chi', \chi''}(G, \zH')$ is the set of all elements of 
$\HOM_{\chi}(G, \zH'')$ that have $(\chi', \chi'')$ as their consistent 
partition. Thus
\begin{align}
\HOM_{\chi}(G, \zH'') = \bigcup_{(\chi', \chi'') \in \cC(\chi)} 
\HOM_{\chi', \chi''}(G, \zH').
\label{recoloring_set}
\end{align}

Equation~(\ref{recoloring-alpha-product}) follows from 
(\ref{recoloring_set}).
\end{proof}

\subsubsection*{Putting things together}

We are now in a position to prove 
Lemma \ref{lemma-beta-operation} and 
Theorem \ref{main-theorem}.

\begin{proof}[Proof of Lemma \ref{lemma-beta-operation}]
By Lemmas~\ref{lemma-H'-H-sum} and~\ref{lemma-H''-H'-sum}, 
\begin{align*}
\hom_{\gm}(G, \zH'') = & \sum\limits_{(\gm', \gm'') \in \cC(\gm)}|
\HOM_{\gm', \gm''}(G, \zH')| \\
= & \sum\limits_{(\gm', \gm'') \in 
\cC(\gm)} |\cB(\gm', \gm'')| \times \hom_{\gm'\uplus \gm''}(G, \zH),
\end{align*}
where, the first equality follows from Lemma~\ref{lemma-H''-H'-sum} 
and the second equality follows from Lemma~\ref{lemma-H'-H-sum}.

The cardinality of $\cC(\gm)$ is bounded by  $2^{|V(G)|}$. Also, for any 
$(\gm', \gm'') \in \cC(\gm)$ the cardinality of $\cB(\gm', \gm'')$ is 
bounded by $(2k)^{V(G)}$ and thus that cardinality is computable in 
time  $O((2k)^{V(G)})$ by brute force. Therefore, the time 
complexity of calculating 
\[
\sum\limits_{(\gm', \gm'') \in \cC(\gm)} 
|\cB(\gm', \gm'')| \times \hom_{\gm'\uplus \gm''}(G, \zH)
\] 
is $O^*((4k)^{|V(G)|})$ and so is the time complexity of computing 
$\hom_{\gm}(G, \zH'')$.  
\end{proof}  

\begin{proposition}\label{homo-consistent-with-partition}
Let $G$ be a graph, $\zH$ a $k$-labeled graph, and let $\Phi$ 
be an extended $k$-expression for $\zH$. Then, there is an 
algorithm that calculates $\hom_{\chi}(G, \zH)$ for all 
$k$-label functions 
$\chi: X \to [k]$, where $X \subseteq V(G)$, in total time 
\[
\size(\Phi)\times (2(k+1))^{ 2|V(G)|},
\]
given numbers
$\hom_{\chi'}(G, \zH')$ for all label functions $\chi'$ and 
$k$-labeled graphs represented by subexpressions of $\Phi$.
\end{proposition}

\begin{proof} 
The base case of induction is when $\size(\Phi) = 1$ which takes 
place when $\zH$ is a labeled vertex. In this case, there is a 
brute force algorithm for computing $\hom_{\chi}(G, \zH)$, with 
running time of $O(|V(G)|)\times (k + 1)^{|V(G)|}$ that is less than 
$\size(\Phi) (2(k+1))^{2|V(G)|}$.

If $\size(\Phi) > 1$, as $\Phi$ is an extended $k$-expression, 
there are three possibilities:
\begin{itemize}
\item 
$\Phi = \beta_{\vec{n}, \sigma, \cS}(\Phi')$, where  $\Phi'$ 
is another extended $k$-expression;
\item 
$\Phi =  \eta_{\cT}(\Phi', \Phi'')$, where  $\Phi'$ and $\Phi''$ are 
two extended $k$-expressions;
\item 
$\Phi = \Psi(\Phi')$, where $\lambda$ is a sequence of 
relabelling operators and $\Phi'$ is another extended $k$-expression 
that does not begin with a relabelling operator. 
\end{itemize}
 
The $k$-labeled graphs represented by $\Phi',\Phi''$ are denoted 
$\zH',\zH''$. Let start with the first case. Let 
$\zH = \beta_{\vec{n}, \sigma,\cS}(\zH')$. Since  $\size(\Phi')$ is 
less than $\size(\Phi)$, by induction hypothesis, there is an algorithm 
that calculates $\hom_{\chi}(G, \zH')$ for all $k$-label functions 
$\chi: X \to [k]$, where $X \subseteq V(G)$, in total time 
$\size(\Phi') \times (2(k+1))^{ 2|V(G)|}$.

Suppose that for all $k$-label functions $\chi: X \to [k]$ such that 
$X \subseteq V(G)$ the value of $\hom_{\chi}(G, \zH')$ is known. 
Then by Lemma~\ref{lemma-beta-operation}, for any $k$-label 
function $\chi:X \to [k]$, where $X \subseteq V(G)$, the value  
$\hom_{\chi}(G, \zH)$ can be found in time $(2k)^{|V(G)|}$ 
and, since there are $(k + 1)^{V(G)}$ of those functions,  
$\hom_{\chi}(G, \zH')$ can also be found for all $k$-label 
functions $\chi: X \to [k]$, where $X \subseteq V(G)$, in time 
$(k + 1)^{|V(G)|} \times (2(k+1)^{|V(G)|})\le (2(k+1))^{2|V(G)|}$. 

So, $\hom_{\chi}(G, \zH)$ for all $k$-label functions 
$\chi:X \to [k]$, where $X \subseteq V(G)$, can be found  in total 
time  $O((\size(\Phi') + 1) \times (2(k+1))^{|V(G)|})$, that is, 
bounded by $\size(\Phi) \times (2(k+1))^{2|V(G)|}$.

The proof for the second and third cases is similar with 
Lemma~\ref{lemma-join-operation} and Lemma~\ref{relabel_lemma}
in place of Lemma~\ref{lemma-beta-operation}.
\end{proof}

\begin{proof}[Proof of Theorem \ref{main-theorem}]
By Observation~\ref{exhaustive-observation}, $\hom(G, H)$  
equals the sum of $\hom_{\chi}(G, \zH)$ over all $k$-label 
functions $\chi: V(G) \to [k]$. 
By Proposition~\ref{homo-consistent-with-partition}, the time 
required to compute  $\hom_{\chi}(G, \zH)$ for all $k$-label functions 
$\chi: V(G) \to [k]$ is $O(\size(\Phi) \times (2(k+1))^{2|V(G)|})$. 
Also, the time for summing up over $k^{V(G)}$ numbers is 
$O(k^{V(G)})$. Thus the total running time can be bounded by 
\begin{align*}
O(\size(\Phi) \times (2(k+1))^{2|V(G)|}) + O(k^{V(G)}) = 
O(\size(\Phi) \times (2(k+1))^{2|V(G)|}).
\end{align*}
\end{proof}

\section{Other examples of plain exponential 
classes}\label{sec:examples}

In this section we provide two 
plain exponential classes of graphs. We start with subdivisions of 
cliques. 

\subsection{Subdivided Cliques } \label{sec:subdivided}

The \emph{subdivision} of an edge $uv$ by a graph $H$ is a 
graph with vertex set $V(H) \cup \{u, v\}$ and edge set 
$E(H) \cup \bigcup\limits_{t \in V(H)} \{ ut,vt\}$. 
The subdivision of a graph $G$ by a graph $H$ is the graph 
obtained by replacing every edge $uv$ of $G$ with its subdivision 
by a copy of $H$ (a disjoint copy for each edge).
Let $\cK(\cH)$ denote the class of subdivisions 
of cliques by graphs from a class $\cH$.

The following theorem is the main result of this section.

\begin{theorem}\label{main-theorem-subdivided}
Let $\cH$ be a plain exponential class of graphs. Then $\cK(\cH)$ 
is also plain exponential. 

More precisely, if {\sc \#GraphHom$(-,\cH)$} can be solved in time 
$O^*(c^{|V(G)|+|V(H)|})$, $c$ constant, for any given graphs 
$G$ and $H\in\cH$, then {\sc \#GraphHom$(-,\cK(\cH))$} can 
be solved in time $O^*(c_1^{2(|V(G)|+|V(H)|})$, where 
$c_1=\max(c,2)$. 
\end{theorem}

The following theorem  from \cite{Koivisto06:algorithm} is 
our main technical tool.
 This theorem gives an upper bound on time complexity of the 
generic problem of summing over the partitions of $n$ elements 
into a given number of weighted subsets. We translate our counting 
problem as summation of this kind and use 
Theorem~\ref{sum-over-partitions} to prove an upper bound for 
time complexity of our counting problem. 

\begin{theorem}[\cite{Koivisto06:algorithm}]\label{sum-over-partitions}
Let $M$ be a $n$-element set and $f$ a function from the power 
set of $M$ to the set of real numbers. Also, suppose that for any 
subset $X\sse M$, the value of $f(X)$ can be obtained in time $O(1)$. 
Let 
\begin{align}
\parr(f) = \sum\limits_{P} \prod\limits_{i = 1}^{n}f(P_i)
\end{align}
where, $P = (P_1, ..., P_n)$ runs through all ordered partitions of 
$M$ into $k$ disjoint subsets of $M$ ($k$ fixed). There is an 
algorithm that calculate $\parr(f)$ in time $O^*(2^{|M|})$.  
\end{theorem}

With specific choices of function $f$ this class of partition problems 
expresses various counting problems. For instance, to compute the 
number of $n$-colourings of a given graph $G$ we let $f(S) = 1$ if the 
vertex subset $S$ is an independent set and $0$ otherwise. As is 
easily seen, for any $n$-partition $P$ of $V(G)$ such that for each 
$i \in [n]$, $P_i$ is an independent set, $P$ can be associated with 
a proper $n$-colouring of $G$ and vice versa. Also $\parr(f)$ is equal to 
the number of $n$-partitions of $V(G)$ that  $P_i$ is an independent 
set for each $i \in [n]$. Thus,  the number of $n$-colouring of $G$  
equals $\parr(f)$. 

Here, we count the number of $H$-colorings of a graph $G$ by 
expressing this number as $\parr(f)$ for some function $f$.
Without loss of generality, we can assume that $G$ is a connected 
graph as otherwise, the number of homomorphisms from $G$ to 
$H$ is the product of the numbers of homomorphisms from the 
connected components of $G$ to $H$. 

Let $H$ be the subdivision of $K_n$ by a graph $U\in\cH$. We 
denote the set of all vertices in $H$ that are originally from 
$V(K_n)$ by $H_A = \{ v_1, ..., v_n\}$ and the set of the remaining  
vertices of $H$ by $H_B = \bigcup_{i,j\in [n], i \neq j } 
\{V(U_{ij})\}$, 
where $U_{ij}$ is the copy of $U$ subdividing the edge $v_iv_j$. 
Let $(A, B)$ be a $2$-partition of $V(G)$. A homomorphism $\vf$ 
from $G$ to $H$ is said to be \emph{consistent} with $(A, B)$ if the 
image of $A$ under $\vf$ is a subset of $H_A$ and the image of 
$B$ under $\vf$ is a subset of $H_B$. 

Let $C_B$ be the set of all connected components of $G[B]$.  
Also let $\bC\in C_B$ and let $N(\bC)$ be the set of all elements 
of $A$ that are adjacent to at least one vertex from $\bC$.  
Unless $A=\emptyset$, the set $N(\bC)$ is not empty as otherwise 
it contradicts the 
assumption that $G$ is connected. Then we fix an arbitrary 
vertex of $N(\bC)$ and refer to it as $s_\bC$. 

Let $\vf$ be an $H$-colouring of $G$. Then it is easy to see 
that vertices from $G$ that are mapped by $\vf$ to vertices of 
$H_A$ should be an independent set. Thus we use brute 
force to consider all 2-partitions $(A, B)$ of $V(G)$, where 
$A$ is an independent set and we  count the number of 
homomorphisms from $G$ to $H$ that are consistent with 
each $(A, B)$ and then sum those numbers up to get the 
total number of homomorphisms from $G$ to $H$. 

Let $\bC \in C_B$. Since any homomorphism $\vf$ consistent 
with $(A, B)$ should map  $\bC$ to $H_B$ and the image of 
$\bC$ under $\vf$ should be connected, the image of $\bC$ under 
$\vf$ is a subset of one of the copies of $U$. 
Also, the image under $\vf$ of all vertices in $A$ adjacent to $\bC$ is 
a set of at most two elements as otherwise it is not possible 
for image of $\bC$ to be adjacent to the image of  
all elements of $N(\bC)$. Thus, one can associate $\vf$ with 
a relation denoted $\vf'\sse (A\cup C_B)\times V(H_A)$ that 
contains pairs $(a,\vf(a))$, $a\in A$, and $(\bC,b)$, where
$\bC \in C_{B}$ and $b\in\vf(N(\bC))$.
%
%
Relation $\vf'$ defined this way may not be a function 
because it may relate a component $\bC$ to more than one 
element of $V(H_A)$. To fix this, we introduce two 
copies of $\bC$  denoted $\bC^0$ and $\bC^1$. 
Since the image of $N(\bC)$ under $\vf$ contains at most 2 
elements, $\vf'$ relates each copy of $\bC$ to one of those 
elements of the image of $N(\bC)$ under $\vf$.  
Now, let $S =A\cup\bigcup_{\bC \in C_B}\{\bC^0, \bC^1\}$, 
 i.e.,  $S$ includes the vertices of $A$ along with two copies of 
 elements of $C_B$. Thus we can associate homomorphisms 
 from $G$ to $H$ that are consistent with $(A, B)$ with 
a relation from $S$ to $V(H_A)$.


In order to convert $\vf'$ into a mapping, for any $\vf$  we 
require the associated $\vf'$ to associate $\bC^1$ and $s_{\bC}$ 
with the same element of $V(H_A)$ and $\bC^0$ with the other 
element if the image of $N(\bC)$ contains two vertices. Restricted
this way $\vf'$ is a mapping $\vf'\colon S\to V(H_A)$.

As any mapping from $S$ to $V(H_A)$ naturally corresponds 
to a $n$-partition of $S$, we can associate  homomorphisms 
from $V(G)$ to $V(H)$ that are consistent with $(A, B)$ with 
$n$-partitions of $S$.
For an $n$-partition $P=(\vc Pn)$ of $S$, we define 
\textit{mappings compatible} with $P$ as follows. 
A mapping $\vf:V(G)\to V(H)$ is compatible with $P$ if it 
satisfies the following properties: 
\begin{itemize}\setlength{\itemindent}{.3cm}
\item[(A)]  
$\vf$ maps elements of $P_i \cap A$ to $v_i$ for each $i \in [n]$.
\item[(B)] 
For every $\bC \in C_B$, such that  $\bC^0 \in P_i$ and 
$\bC^1 \in P_j$, for $i \neq j$, the mapping $\vf\vert_\bC$ is a 
homomorphism that maps $\bC$ to $U_{ij}$.
\item[(C)] 
For every $\bC \in C_B$ such that $\bC^0, \bC^1 \in P_i$, 
$\vf\vert_\bC$ is a homomorphism that maps  $\bC$ to $U_{ij}$, 
where $j \in [n] - \{i\}$ is arbitrary.
\end{itemize}

We define function $f$ from the power set of $S$ to natural 
numbers as follows.
First, $f$ is set to $0$ for subsets $X\sse S$ 
such that if $X$ is a class of a partition $P$, say, $X=P_i$,
then there are no homomorphisms compatible with $P$.
So suppose  $P$  has at least one homomorphism compatible
with it. Then for any $\bC \in C_B$, the class of $P$ that 
contains $a \in N(\bC)$ also contains one of $\bC^0$ 
or $\bC^1$. So, for $X \subseteq S$,
\begin{quote}
$f(X) = 0$ if there is $a \in X$ and $ \bC \in C_B$ such that 
$a \in N(\bC)$  but neither of  $\bC^0$ or $\bC^1$ is in the 
set  $X$.
\end{quote}

Also, if $P$ contains at least one of $\bC^0$ 
or $\bC^1$ in one of its classes for some $\bC \in C_B$, 
then there should be at least one of vertex from 
$N(\bC)$ in the same subset as otherwise $P$ would 
have no compatible homomorphisms. Thus for  $X\subseteq S$, 
\begin{quote}
 $f(X) = 0$ if for some $\bC \in C_B$,  $\bC^1  \in X$ 
 and/or $\bC^0 \in X$ but none of the elements of $N(\bC)$ is in $X$. 
\end{quote}
Recall that the $n$-partition compatible with a homomorphism 
is supposed to have $\bC^1$ and $s_{\bC}$ in the same class 
of $P$. Therefore, for $X \subseteq S$,
\begin{quote}
 $f(X) = 0$ if for some $\bC \in C_B$, $\bC^1 \in X$ and 
 $s_{\bC} \not \in X$. 
\end{quote}
Now we need to set the value of $f$ for subsets that 
can be a class in a partition with compatible homomorphisms.
For $X\subseteq S$, 
\begin{itemize}
\item 
if none of the above is true for $X$, $f(X) = x\cdot y$, where 
\begin{align*}
x &= \prod\limits_{\bC: \bC^0, \bC^1 \in X} (n - 1) 
\hom(\bC, U) \textit{, and} \\
y &= \prod\limits_{\bC : \bC^1\in X, \bC^0 \not \in X }\hom(\bC, U).
\end{align*}
\end{itemize}
We defer an explanation of this equality to the proof of the following 

\begin{lemma}\label{homo-partition-relation}
The number of homomorphisms from $G$ to $H$ that 
are consistent with $(A, B)$ equals
\begin{align}
\sum_P \prod_{i = 1}^{n} f(P_i),
\end{align}
where $P = (P_1, ..., P_n)$ runs through all $n$-partitions of $S$. 
\end{lemma}

\begin{proof}
To prove this lemma,  we prove that the $n$-partitions of $S$ 
induce a partition on the set of homomorphisms from 
$G$ to $H$ that are consistent with $(A, B)$, by grouping all 
homomorphism according to their compatible $n$-partitions. 
Then, we show that for each $n$-partition 
$P$ of $S$, the number of homomorphisms compatible with 
$P$ equals $\prod_{i = 1}^{n} f(P_i)$. These two claims prove 
the lemma.

Let  $P$ be a fixed $n$-partition of $S$ such 
that $\prod_{i = 1}^{n} f(P_i) \neq 0$, and let $\vf$ 
be a mapping compatible with $P$. 
As is easily seen, $\vf$ is a mapping from $G$ to $H$ that is 
consistent with $(A, B)$. Since $A$ is an independent set, 
$\vf\vert_{A}$ is a homomorphism for any mapping $\vf$. By 
Properties~(B) and~(C), for any connected component 
$\bC \in C_B$, $\vf\vert_{V(\bC)}$ is a homomorphism 
hence $\vf\vert_{B}$ is a homomorphism as well. Now 
let $a \in A$ and $b\in V(\bC)$, where $\bC\in C_B$. Since 
$\prod_{ i = 1}^{n} f(P_i) \neq 0$
the possible cases are: (1) $a, \bC^0, \bC^1 \in P_i$ for 
some $i \in [n]$; (2) $a, \bC^1 \in P_i$ and  $\bC^0 \in P_j$ 
such that $i \neq j$, if $a=s_\bC$, and $\bC^1 \in P_i$ and  
$a,\bC^0 \in P_j$ otherwise. In case (1) 
by Property 3 $\vf(\bC)$ is 
in graph $U_{ij}$ for some $0 \leq j \leq n$, and by 
Property~(A),  $\vf(a)$ is $v_i$, and since there is an edge 
between any vertex of $U_{ij}$ and $v_i$, $\vf(a)\vf(b)$ 
is an edge of $H$. The argument in case (2) goes in 
the similar way. So, any mapping compatible with partition 
$P$ is a homomorphism that is consistent with $(A, B)$. 

Now, we prove that for any homomorphism $\vf$ from $G$ to $H$ 
that is consistent with $(A, B)$, there is exactly one 
$n$-partition $P$ of $S$ such that $\vf$ is compatible with $P$. 
Let us fix a homomorphism $\vf$ from $G$ to $H$ that is 
consistent with $(A, B)$. By Properties(A)--(C) of the 
compatible homomorphism, an $n$-partition $P$ with 
which $\vf$ is compatible should satisfy the following properties 
\begin{itemize}
\item[(a)] 
For any $a \in A$, $a \in P_i$ if and only if $\vf(a) = v_i$.
\item[(b)] 
For any $\bC \in C_B$ such that the image of $N(\bC)$ 
under $\vf$ is $\{ v_i, v_j\}$ for $1 \leq i \neq j \leq n$, and 
$\vf(s_{\bC}) = v_i$ it holds $\bC^0 \in P_j$ and 
$\bC^1 \in P_i$. Note that $\bC^0$ and $\bC^1$ are 
not interchangeable because $s_{\bC}$ and $\bC^1$ 
should be included in the same class of $P$.   
\item[(c)] 
For any $\bC \in C_B$ such that the image of 
$N(\bC)$ under $\vf$ is $\{v_i\}$, it holds 
that $\bC^0,\bC^1\in P_i$.  
\end{itemize}

As is easily seen, there is at least one $n$-partition 
that satisfies all of the above properties. On the other hand, 
no two $n$-partitions of $S$ can satisfy all the above properties 
because the properties specify for each of elements of 
$S$  which subset of the compatible partition to go, 
uniquely. Thus, there is exactly one $n$-partition of $S$ 
that is compatible with $\vf$.

Let $P$ be an $n$-partition of $S$ and $\hat{P}$ be the 
set of all homomorphisms compatible with $P$. 
We show that $|\hat{P}| = \prod\limits_{i = 1}^{n} f(P_i)$. 
For any connected component $\bC\in C_B$, such that 
$\bC^0 \in P_i$ and $\bC^1 , s_{\bC} \in P_j$ for 
$i \neq j$, $\hom(\bC, U)$ is a contributing factor to $f(P_j)$. 
Note that $\hom(\bC, U)$ does not contribute anything 
to $f(P_i)$. Also, for any connected component $\bC\in C_B$, 
such that $\bC^0,\bC^1 \in P_i$ for some 
$i \in [n]$, the value $f(P_i)$ contains a factor 
$(n - 1)\hom(\bC, U)$.  Thus, 
\begin{align}
\prod\limits_{i = 1}^{n} f(P_i) = 
\prod\limits_{\substack{\bC: 
\bC^0, \bC^1 \in P_i \\ \text{for some $i \in [n]$} }} 
(n - 1)\hom(\bC, U)
\times \prod\limits_{\substack{\bC : 
\bC^1\in P_i, \bC^0 \not \in P_i \\ 
\text{for some $i \in [n]$}} }\hom(\bC, U).
\label{multiplication}
\end{align}

On the other hand by Property~(a), $\vf\vert_{A}$ is the same 
for every $\vf\in \hat{P}$.  Let $\bC \in C_B$.  Let 
$\vf_1, \vf_2 \in \hat{P}$ be two  homomorphisms  such that 
$\vf_1 \vert_{V(\bC)}$ and $\vf_2 \vert_{V(\bC)}$ are not 
identical. Then, let  $\vf: V(G) \to V(H)$ be defined as follows:
\begin{align*}
\vf(v) = 
\begin{cases}
\vf_1(v), & \textit{if $v \in V(G) \setminus V(\bC)$,}\\
\vf_2(v), & \textit{if $v \in V(\bC)$.}
\end{cases}
\end{align*}
Mapping $\vf$ is a homomorphism as $\vf_1 \vert_A$ and 
$\vf_2 \vert_A$ are the same and also $\bC$ is disconnected from the 
$V(G) \setminus (V(\bC) \cup A)$.  Thus, the number of compatible 
homomorphisms of partition $P$ can be obtained by multiplying 
the number of all different mappings we can get by restricting 
those homomorphisms to a component $\bC \in C_B$, over all  
$\bC \in C_B$. Therefore, 
\begin{align*}
|\hat{P}| &= \prod_{\bC \in C_B } \quad\Bigm\lvert 
\{\vf \vert_{V(\bC)}:\ \ \vf \in \hat{P}\}\Bigm\lvert \\
&= \prod_{\substack{\bC: \bC^0, \bC^1 \in P_i \\ 
\text{for some $i \in [n]$} }} \quad\Bigm\lvert 
\{\vf \vert_{V(\bC)}:\ \ \vf \in \hat{P}\}\Bigm\lvert
\times \prod_{\substack{\bC : \bC^1\in P_i, \bC^0\not\in P_i \\ 
\text{for some $i \in [n]$}}} \quad\Bigm\lvert 
\{\vf \vert_{V(\bC)}:\ \ \vf \in \hat{P}\}\Bigm\lvert.
\end{align*}

Let again $\bC \in C_B$, such that $\bC^0 \in P_i$ and 
$\bC^1 \in P_j$ 
for $i \neq j$. By Property~(a) every $\vf\in \hat{P}$ maps 
$P_i \cap A$ to $v_i$ and since $\vf$ is homomorphism, it maps 
$C$ to $U_{ij}$. Therefore, 
\[
\Bigm\lvert\{\vf \vert_{V(\bC)}:\ \ \vf\in \hat{P}\}\Bigm\lvert=
\hom(\bC, U).
\] 
Now let $\bC \in C_B$ such that $\bC^0 , \bC^1 \in P_i$. 
By Property~(C), the image of $\bC$ under $\vf$ can be any of 
$U_{ij}$ where $j \in [n] - \{i\}$. Thus, 
\[
\Bigm\lvert \{\vf\vert_{V(\bC)}:\ \ \vf\in \hat{P}\} 
\Bigm\lvert = (n - 1)\hom(\bC, U).
\] 
Therefore,
\begin{align}
|\hat{P}| = \prod_{\substack{\bC: \bC^0, \bC^1 \in P_i \\ 
\text{for some $i \in [n]$} }} (n - 1)\hom(\bC, U)
\times \prod_{\substack{\bC : \bC^1\in p, \bC^0 \not \in P_i \\ 
\text{for some $i \in [n]$}}}\hom(\bC, U)
\label{sizee-of-B}
\end{align}

By (\ref{sizee-of-B}) and~(\ref{multiplication}), 
$|\hat{P}| = \prod\limits_{i= 1}^{n} f(P_i) $.

Hence, the number of homomorphisms from $G$ to $H$ 
that are consistent with $(A, B)$ is equal to 
$\sum\limits_{P}\prod\limits_{i = 1}^{n}f(P_i)$, where 
$P = (P_1, ..., P_n)$ runs through all $n$-partitions of $S$.
\end{proof}

Now, we are in a position to prove 
Theorem~\ref{main-theorem-subdivided}

\begin{proof}[ Proof of Theorem \ref{main-theorem-subdivided}]
We brute force over all partitions of set $V(G)$ into two sets $A$ 
and $B$, and for each partition $(A, B)$ of $V(G)$, we calculate 
function $f$ with respect to $(A,B)$ and then we find the $\parr(f)$. 
Let us fix a partition $(A, B)$ of $V(G)$.  The time needed to 
calculate function $f$ with respect to $(A, B)$, over all subsets 
of $S$ is bounded by $O^*(c^{|V(G)|})$ as the time to calculate 
$\hom(\bC, U)$ is bounded by $O(c^{|V(C)|})$ and there are at 
most $|V(G)|$ of them. By Lemma~\ref{homo-partition-relation} 
the number of homomorphisms from $G$ to $H$ that are 
consistent with $(A, B)$ can be represented as $\parr(f)$ for 
the corresponding $f$.  
By Theorem~\ref{sum-over-partitions}, the total time to calculate 
$\parr(f)$ is bounded by  $O^*( 2^{|V(G)|})$. Thus, the time 
complexity to calculate the number of homomorphisms from 
$G$ to $H$ that are consistent with $(A, B)$ is bounded by 
$O^*(c_1^{|V(G)|})$, $c_1=\max(c,2)$. 

The value $\hom(G, H)$ equals the sum of numbers of 
homomorphisms from $G$ to $H$ that are consistent with $(A, B)$ 
over all $2$-partitions $(A, B)$ of $V(G)$. Since there are $2^{V(G)}$ 
different partitions $(A, B)$ and for any such partition the time 
complexity of counting the number of homomorphisms from 
$G$ to $H$ is bounded by $O^*(c_1^{|V(G)|})$, there is a algorithm 
that computes $\hom(G, H)$ in time $O^*(c_1^{2|V(G)|})$. 
\end{proof}

By Theorem 5.4 from \cite{Makowsky99:clique} it follows that 
$\cK(\cH)$ does not always have bounded clique width;
for instance when $\cH$ contains only one graph that is a single vertex.

\subsection{Kneser Graphs}\label{sec:kneser}

Kneser graphs give another example of a plain exponential class 
of graphs.

The \textit{Kneser graph} $\KG_{n, k}$ is the graph whose vertex 
set is the set of $k$-element subsets of a set of $n$ elements, and 
two vertices are adjacent if and only if the two corresponding sets 
are disjoint. By $\KGG_k$ we denote the class of all Kneser 
graphs for a fixed $k$. 
The main result of this section is the following 

\begin{theorem}\label{the:kneser}
For any $k$, $\KGG_k$ is plain exponential.

More precisely, there is an algorithm that given a graph $G$ and 
$n\in\mathbb N$ finds $\hom(G,\KG_{n, k})$ in time 
$O^*(2^{k|V(G)|})$. 
\end{theorem}

\begin{remark}
The running time of our algorithm is not plain exponential if 
$k$ is not a constant. Thus, Theorem~\ref{the:kneser} does not 
guarantee that the class of all Kneser graphs irrespective of the 
parameter $k$ is plain exponential. Proving or disproving it remains
an open problem. 
\end{remark}

We use the following fact.

\begin{theorem}[\cite{Koivisto06:algorithm}]\label{lem:k-colouring}
There is an algorithm that given $k$ and a graph $G$ counts the 
number of $k$-colouring of graph $G$ in time $O^*(2^{|V(G)|})$.
\end{theorem}

\begin{proof}[Proof of Theorem~\ref{the:kneser}]
Let $G^{(k)}$ denote the graph obtained by replacing  each 
of its vertices with a clique of size $k$ and replacing each of 
its edges with a complete bipartite graph on $k + k$ vertices. 
For $a \in V(G)$ let $\psi(a)$ denote the set of vertices of the 
clique replacing $v$ in $G^{(k)}$. 

First, we introduce a many to one correspondence between 
elements of $\HOM(G^{(k)}, K_n)$ and $\HOM(G,\KG_{n, k})$. 
Let $\tau:\HOM(G^{(k)}, K_n)\to\HOM(G,\KG_{n, k})$ be defined
by setting $\tau(\vf): V(G) \to\KG_{n, k}$ to be the mapping
$v \mapsto \{\vf(u) | u \in \psi(v) \}$.
Notice that the cardinality of $\{\vf(u) | u \in \psi(v) \}$ equals 
$k$ because $G^{(k)}[\psi(v)]$ is a $k$-clique and $\vf$ is a 
homomorphism from $G^{(k)}$ to $K_n$. Therefore $\tau(\vf)(v)$
is always a vertex of $\KG_{n, k}$.

Next, we prove that for $\vf\in\HOM(G^{(k)}, K_n)$, 
$\vf^*=\tau(\vf)$ is  a homomorphism. Let $ab\in E(G)$ then 
$\vf^*(a) =  \{\vf(u) | u \in \psi(a) \}$ and 
$\vf^*(b)= \{\vf(u)\mid u\in \psi(b)\}$. Since 
$G^{(k)}[\psi(a) \cup \psi(b)]$ is a clique, and $\vf$ is a 
homomorphism $\vf^*(a)$ and $\vf^*(b)$ are two disjoint set 
thus $\vf^*(a)\vf^*(b))$ is an edge of $\KG_{n, k}$. Therefore 
$\vf^*\in\HOM(G,\KG_{n, k})$. 

Next, we prove that for any element $\sg$ of $\HOM(G,\KG_{n,k})$,  
$\tau(\vf) =\sg$ for exactly $(k!)^{|V(G)|}$ homomorphisms 
$\vf\in\HOM(G^{(k)}, K_n)$. Take the set of all   
$\vf\in\MAP(G^{(k)}, K_n)$ 
such that for any $v \in V(G)$ the image of $\psi(v)$ under $\vf$ 
equals $\sg(v)$. First, there are $(k!)^{|V(G)|}$ of them as for 
a fixed $\sg$ and a fixed $v \in V(G)$, the image of any $G[\psi(v)] $ 
under such a mapping, has $k!$ different possibilities.
Second, observe that for any $\vf\in\MAP(G^{(k)}, K_n)$ such 
that for any $v \in V(G)$ the image of $\psi(v)$ under $\vf$  
equals $\sg(v)$, $\vf$ is a homomorphism, we have 
$\tau(\vf)=\sg$. 
Hence, $\vf$ is a surjective $(k!)^{|V(G)|}$ to $1$  mapping from 
$\HOM(G^{(k)}, K_{n})$ to $\HOM(G,\KG_{n, k})$. Therefore 
\[
|\HOM(G,\KG_{n, k})| = \frac{\HOM(G^{(k)}, K_n)}{(k!)^{|V(G)|}}.
\] 
Since there is an algorithm that computes $\hom(G^{(k)}, K_n)$ in 
time $O^*(2^{k|V(G)|})$, there is an algorithm that computes 
$\hom(G,\KG_{n, k})$ in the same time. 
\end{proof}

Next we show that $\KGG_k$ does not have bounded clique 
width. 

\begin{theorem}\label{lem:kneser-unbounded}
The class $\KGG_2$ of Kneser graphs does not have bounded 
clique width. 
\end{theorem}

The next theorem is the main ingredient.

\begin{theorem}[\cite{Blanche17:cliquewidth}]\label{inverse-bounded}
Let $\cH$ be a class of graphs with bounded clique width. Then, 
the class of complements of graphs from $\cH$ is also of bounded clique 
width.
\end{theorem}

\begin{proposition}\label{inverse-not-bounded}
The class of complements of graphs of $\KGG_2$ does not have 
bounded clique width. 
\end{proposition}

\begin{proof}
Let $k$ be a constant, let $n$ be a sufficiently large number, 
let $V$ be a set of cardinality $n$, and, for the sake of contradiction, 
let  $\Phi_0$ be a $k$-expression for the complement of a Kneser graph 
$\KG_{n,2}$ whose vertices are 2-element subsets of $V$. 
In particular, this means that each vertex of $\KG_{n,2}$ receives 
one of the $k$ labels. We 
define a finite sequence of $k$-expressions $\{\Phi_i\}$  
which starts with $k$-expression $\Phi_0$.

For $i \ge 0$, the $( i + 1)$-th element of the sequence is defined 
from the $i $-th element of the sequence as follows: if $\Phi_i$ 
represents a graph with more than 1 vertices, it has the form 
$\Psi_i(\Phi_{i,l} \bigoplus \Phi_{i,r})$, where 
$\Psi_i$ is a sequence of recolouring and connecting operators 
and $\Phi_{i,l}$ and $\Phi_{i,r}$ are two $k$-expressions. Then, 
let $\Phi_{i + 1}$ be the $k$-expression from $\{\Phi_{i,l},\Phi_{i,r}\}$
that represents the graph with greatest number of vertices.
Thus, if $G_i$ is the graph represented by $\Phi_i$ and $G_{i+1}$
the graph represented by $\Phi_{i+1}$, then 
$|V(G_i)| \geq |V(G_{i - 1})|/2$. 

Let graph $G$ be a subgraph of $G_0$. For $v \in V$, let 
\textit{count} of $v$ in graph $G$, denoted by $c(v)$, be the 
number of vertices from $G$ of which $v$ is an element. 
Note that each vertex of $G_0$ is a $2$-elements set of elements 
in $V$. Let $G_i$ be the first graph in the sequence such that the 
number of vertices whose count in $G_i$ is $n - 1$ is less 
than $c_0 = (n - 1)/ k $. There are two cases: first, in $G_i$ the
count of every vertex is less than $n - 1$, and second, there is 
at least one vertex in $V$ whose count in $G_i$ is $n - 1$. Also 
note that the count of any $v \in V$ in $G_0$ or any of its subgraphs, 
is less than $n$.

 Let $A$ denote the set of all elements of $V$ whose count  
 is $n - 1$ in $G_{i - 1}$. Then, $|A| \geq c_0$. 
For each element of $A$ there 
are $n - 1$ different vertices of $G_{i - 1}$ which contain it. Thus 
 the sum of $c(v)$ over all $v \in A$ is $(n - 1) |A| $ which is number 
 of all vertices of $G_{i - 1}$ that have one of their elements in $A$ and 
 the other element in $V \setminus A $, plus twice  
 the number of vertices of $G_{ i - 1}$ that have both elements in 
 $A$. Thus the number of vertices of $G_{i - 1}$ that have at least 
 one element in $A$ equals $(n - 1)|A| - \binom{|A|}{2}$. 
Observe that $|V(G_{i - 1})|$ is greater than the number of 
 those vertices of $G_{i - 1}$ that contain an element from $A$. 
 As $\frac{n - 1}{k} \leq |A| \leq n$,
\begin{align*}
|V(G_{i - 1})| \geq (n - 1)|A| - \binom{|A|}{2} = 
|A|\left(n - 1 - \frac{|A|  - 1}{2}\right) \geq \frac{(n - 1)^2}{2k}.
\end{align*}

Therefore 
\begin{align*}
|V(G_{i})| \geq  |V(G_{i - 1})| \setminus 2 \geq \frac{ (n - 1)^2}{4k}.
\end{align*}

Let $v_1$ be an element of $V$ whose count in $G_i$ is the 
largest among all elements of $V$. By the Pigeonhole principle, 
\begin{align*}
c(v_1) \geq \frac{|V(G_i)|}{n}   \geq \frac{2 (n - 1)^2}{4nk}.
\end{align*}

Let $G_i(v)$ be the set of all vertices of $G_i$ that have $v_1$ as 
an element. Again, by the Pigeonhole principle, there are at least 
$c(v_1)/k $ vertices of $G_i(v_1)$ that are labeled with 
the same label. So, for sufficiently large $n$ there are at least $3$ 
vertices in $G_i$ that all have $v_1$ as an element and are labeled 
with the same label in $G_i$. Now, let $\{v_1,v_2\}$, 
$\{v_1, v_3\}$, and $\{v_1, v_4\}$ be different elements of 
$G_i(v_1)$  that are labeled with the same label. 

Now, let us consider the case when the count of each 
vertex of $V$ is less than $n - 1$ in $G_i$. In this case there is 
$v_5 \in V$ such that $\{v_2, v_5\}\not \in V(G_i)$. Clearly, 
$v_5\ne v_3$ or $v_5\ne v_4$. Without loss of generality, 
assume $v_5\ne v_3$. 

In the case when there are vertices with count $n-1$ in $G_i$, 
let $v_1 \in V$ be an element whose count 
in $G_i$ is $n - 1$. Then, by the Pigeonhole principle, there are 
$(n - 1)/k$ elements of $G_i(v_1)$ that have the 
same label. Since there are at least $(n - 1)/k$ elements 
of $G_i(v_1)$ with the same label, and there are at most $c_0$ 
elements of $V$ with count $n - 1$ in $G_i$, there are 
$\{v_1, v_2\}, \{v_1, v_3\}, \{v_1, v_4\} \in V(G_i)$, such that 
$\{v_1, v_2\}$, $\{v_1, v_3\}$, and $\{v_1, v_4\}$ are labeled the 
same label in $G_i$ and also count of $v_2$ is less than $n - 1$ in 
$G_i$. Therefore, there is $v_5 \in V$ such that  
$\{v_2, v_5\} \not \in V(G_i)$ and (say) $v_5 \neq v_3$.

Thus, in both cases, there are $v_1, v_2, v_3,v_5 \in V$ such 
that $\{v_1, v_2\}$ and  $\{v_1, v_3\}$ are from  $V(G_i)$, and also 
are labeled with the same label, and $\{v_2, v_5\} \not \in V(G_i)$, 
and $v_5 \neq v_3$. Since $G_i$ is represented by a subexpression 
of $\Phi_0$ and $\{v_2, v_5\}\not \in G_i$, the vertex $\{v_2, v_5\}$ 
should connect to $\{v_1, v_2\}$ at some later point. Since 
$\{v_1, v_2\}$ and $\{v_1, v_3\}$ both have the same label, 
any operator $\eta_{ij}$ that connects $\{v_2, v_5\}$ to 
$\{v_1, v_2\}$, also connects 
$\{v_2, v_5\}$ to $\{v_1, v_3\}$ that is not an edge of $G_0$. 
A contradiction. 
\end{proof}

\bibliographystyle{plain}

\end{document}